\documentclass[a4paper,11pt,dvipsnames]{article}
 
\newcommand{\ignore}[1]{}

\usepackage[x11names, rgb]{xcolor}
\usepackage[utf8]{inputenc}
\usepackage{tikz}
\usetikzlibrary{snakes,arrows,shapes}
\usepackage{fullpage}
\usepackage[colorlinks=false, pagebackref]{hyperref}  
\usepackage{microtype}
\usepackage{textcomp}
\usepackage{amsmath, amsfonts, amssymb, amsthm}
\usepackage{mathtools, mathdots}
\usepackage{mhsetup}
\usepackage{algorithm, algpseudocode, algorithmicx}
\usepackage{enumerate}

\newtheorem{theorem}{Theorem}

\newtheorem{corollary}[theorem]{Corollary}
\newtheorem{lemma}[theorem]{Lemma}
\newtheorem{definition}[theorem]{Definition}
\newtheorem{proposition}{Proposition}

\newtheorem{conjecture}{Conjecture}
\newtheorem{claim}[theorem]{Claim}

\newcommand{\claimproof}[2]%
{\noindent{\em Proof of Claim \ref{#1}.}
#2\hspace*{\fill}$\Box$~~~~\vspace{3.5mm} }

\newcommand{\F}{\mathbb{F}}

\newcommand{\R}{\mathbb{R}}
\newcommand{\Q}{\mathbb{Q}}

\newcommand{\C}{\mathbb{C}}

\newcommand{\poly}{\text{poly}}
\newcommand{\siz}{\text{size}}
\newcommand{\rad}{\text{rad}}

\begin{document}

\pagenumbering{gobble}


\title{\bf Discovering the roots: Uniform closure results for algebraic classes under factoring}

\author{
Pranjal Dutta \thanks{Chennai Mathematical Institute, \texttt{pranjal@cmi.ac.in} }
\and
Nitin Saxena \thanks{CSE, Indian Institute of Technology, Kanpur, \texttt{nitin@cse.iitk.ac.in} }
\and
Amit Sinhababu \thanks{CSE, Indian Institute of Technology, Kanpur, \texttt{amitks@cse.iitk.ac.in} }
}

\date{}
\maketitle

\begin{abstract}
Newton iteration (NI) is an almost 350 years old recursive formula that  approximates a simple root of a polynomial quite rapidly. We generalize it to a matrix recurrence (allRootsNI) that approximates {\em all} the roots simultaneously. In this form, the process yields a better circuit complexity in the case when the number of roots $r$ is small but the multiplicities are exponentially large. Our method sets up a linear system in $r$ unknowns and iteratively builds the roots as formal power series. For an algebraic circuit $f(x_1,\ldots,x_n)$ of size $s$ we prove that {\em each} factor has size at most a polynomial in: $s$ and the degree of the squarefree part of $f$. Consequently, if $f_1$ is a $2^{\Omega(n)}$-hard polynomial then any nonzero multiple $\prod_{i} f_i^{e_i}$ is equally hard for {\em arbitrary} positive $e_i$'s, assuming that $\sum_i\deg(f_i)$ is at most $2^{O(n)}$.

It is an old open question whether the class of poly($n$)-sized formulas (resp.~algebraic branching programs) is closed under factoring. We show that given a polynomial $f$ of degree $n^{O(1)}$ and formula (resp.~ABP) size $n^{O(\log n)}$ we can find a similar size formula (resp.~ABP) factor in randomized poly($n^{\log n}$)-time. 
Consequently, if determinant requires $n^{\Omega(\log n)}$ size formula, then the same can be said about any of its nonzero multiples.

As part of our proofs, we identify a new property of multivariate polynomial factorization. We show that under a random linear transformation $\tau$, $f(\tau\overline{x})$ {\em completely} factors via power series roots. Moreover, the factorization adapts well to circuit complexity analysis. This with allRootsNI are the techniques that help us make progress towards the old open problems; supplementing the large body of classical results and concepts in algebraic circuit factorization (eg.~Zassenhaus, J.NT 1969; Kaltofen, STOC 1985-7 \& B\"urgisser, FOCS 2001).

\end{abstract}

\vspace{-.30mm}
\noindent
{\bf 2012 ACM CCS concept:} Theory of computation-- Algebraic complexity theory,  Problems, reductions and completeness; Computing methodologies-- Algebraic algorithms, Hybrid symbolic-numeric methods; Mathematics of computing-- Combinatoric problems.

\vspace{-.45mm}


\pagenumbering{arabic}

\vspace{-1mm}

\section{Introduction}
\vspace{-1mm}

Algebraic circuits provide a way, alternate to Turing machines, to study computation. Here, the complexity classes contain (multivariate) polynomial families instead of languages. It is a natural question whether an algebraic complexity class is closed under factors. This is also a useful, and hence, a very well studied question both from the point of view of practice and theory.  
We study the following two questions related to multivariate polynomial factorization: 
{\bf (1)} Let $\{f_n(x_1,\ldots,x_n)\}_{n}$ be a polynomial family in an algebraic complexity class $\mathcal{C}$ (egs.~VP, VF, VBP, VNP or $\overline{\text{VP}}$ etc.). Let $g_n$ be an arbitrary factor of $f_n$. Can we say that $\{g_n\}_{n} \in \mathcal{C} $? Equivalently, is the class $\mathcal{C}$ {\em closed under factoring}?
{\bf (2)} Can we design an {\em efficient}, i.e.~randomized poly($n$)-time, algorithm to output the factor $g_n$ with a representation in $\mathcal{C}$? ({\em Uniformity})

Different classes give rise to new challenges for the closure questions. Before discussing further, we give a brief overview of the algebraic complexity classes relevant for our paper. For more details, see \cite{mahajan2014algebraic, shpilka2010arithmetic,burgisser2013algebraic}.

Algebraic circuit is a natural model to represent a polynomial compactly.
An {\em algebraic circuit} has the structure of a layered directed acyclic graph. It has leaf nodes labelled as input variables $x_1,\ldots,x_n$ and constants from the underlying field $\F$. All the other nodes are labelled as addition and multiplication gates. It has a root node that outputs the polynomial computed by the circuit. Some of the complexity parameters of a circuit are
\emph{size} (number of edges and nodes), \emph{depth} (number of layers), syntactic {\em degree} (the maximum degree polynomial computed by any node), \emph{fan-in} (maximum number of inputs to a node) and \emph{fan-out}. An {\em algebraic formula} is a circuit whose underlying graph is a \emph{directed tree}. In a formula, the fan-out of the nodes is at most one, i.e.~`reuse' of intermediate computation is not allowed.

The class VP (resp.\ VF) contains the families of $n$-variate polynomials of degree $n^{O(1)}$ over $\F$, computed by $n^{O(1)}$-sized circuits (resp.~ formulas). The class VF is sometimes denoted as $\text{VP}_e$, for it collects `expressions' which is another name for formulas. Similarly, one can define VQP (resp.\ VQF) which contains the families of $n$-variate polynomials of degree $n^{O(1)}$ over $\F$, computed by $2^{\text{poly}(\log n)}$-sized circuits (resp.\ formulas). If we relax the condition on the degree in the definition of VP, by allowing the degree to be possibly exponential, then we define the class $\text{VP}_{nb}$. Such circuits can compute constants of exponential bit-size (unlike VP).
 
Algebraic branching program (ABP) is another model for computing polynomials which we define in Sec.\ref{prel}.
The class VBP contains the families of polynomials computed by $n^{O(1)}$-sized ABPs. We have the easy containments: VF $\subseteq$ VBP $\subseteq$ VP $\subseteq$ VQP $=$ VQF \cite{ben1992computing, valiant1983fast}.

Finally, we give an overview of the class VNP, which can be seen as a non-deterministic analog of the class VP.
A family of polynomials $\{f_n\}_n$ over $\F$ is in $\text{VNP}$ if there exist polynomials $t(n), s(n)$ and a family $\{g_n\}_n$ in $\text{VP}$ such that for every $n$, $ f_n(\overline{x})=\sum_{w \in \{0,1\}^{t(n)}} g_n(\overline{x},w_1,\hdots,w_{t(n)})$. Here, \emph{witness} size is $t(n)$ and \emph{verifier} circuit $g_n$ has size $s(n)$. VP is contained in VNP and it is believed that this containment is strict (Valiant's Hypothesis \cite{V79}). 

\smallskip
Newton iteration is one of the most popular numerical methods in engineering \cite{ortega2000iterative, gill1986projected}. This work introduces a new process to approximate all the roots of a circuit assuming that they are few and their multiplicites are known. This is based on a matrix recurrence, which in turn is derived from a new identity (Claim \ref{lem-logDer-Id}). Based on the process (called {\em allRootsNI} in Section \ref{sec-tech}) we get several consequences in high-degree circuit factoring (eg.~Theorem \ref{thm1}):

{\em Every factor of a given circuit C has size polynomial in: size(C) and the degree of the squarefree part of C.}

\smallskip\noindent
and in factoring other poly-degree algebraic models (eg.~Theorems \ref{thm3} \& \ref{thm-vf-bar}):

{\em Every factor, of a degree-$d$ polynomial with VF (respectively VBP, VNP) complexity $s$, has VF (respectively VBP, VNP) complexity poly($s,d^{\log d}$). The latter is poly($s$) if degree $d=2^{O(\sqrt{\log s})}$.}

\smallskip
Now, we briefly discuss the state of the art on the closure questions for various algebraic complexity classes.
To cover more depth and breadth, see \cite{kaltofen1990polynomial, kaltofen1992polynomial, forbes2015complexity}.

\vspace{-2mm}
\subsection{Previously known closure results}
\vspace{-1mm} 
 
Famously, Kaltofen \cite{Kaltofen85, Kaltofen86, kaltofen1987single, kaltofen1989factorization} showed that VP is {\em uniformly} closed under factoring, i.e.~for a given $d$ degree $n$ variate polynomial $f$ of circuit size $s$, there exists a randomized poly$(snd)$-time algorithm that outputs its factor as a circuit whose size is bounded by poly$(snd)$. This fundamental result has several applications such as  `hardness versus randomness' in algebraic complexity \cite{kabanets2003derandomizing,agrawal2008arithmetic,  dvir2009hardness, AFGS17}, derandomization of Noether Normalization Lemma \cite{mulmuley2017geometric}, in the problem of circuit reconstruction  \cite{karnin2009reconstruction, sinha2016reconstruction}, and polynomial equivalence testing \cite{kayal2011efficient}. In general, multivariate polynomial factoring has several applications including decoding of Reed-Solomon, Reed-Muller codes \cite{guruswami1998improved,sudan1997decoding}, integer factoring \cite{lenstra1990number}, primary decomposition of polynomial ideals \cite{gianni1988grobner} and algebra isomorphism \cite{KS06, IKRS12}.

   It is natural to ask whether Kaltofen's VP factoring result can be extended to $\text{VP}_{nb}$ which allows degree of the polynomials to be exponentially high. It is known that {\em not every} factor of a high degree polynomial has a small sized circuit. For example, the polynomial $x^{2^s}-1$ can be computed in size $s$, but it has factors over $\mathbb{C}$ that require circuit size $\Omega\left(2^{s/2}/\sqrt{s}\right)$ \cite{lipton1978evaluation,schnorr1977improved}. It is conjectured \cite[Conj.8.3]{burgisser2013completeness} that \emph{low} degree factors of high degree small-sized circuits have \emph{small} circuits. 
   Partial results towards it are known. It was shown in \cite{kaltofen1987single} that if  polynomial $f$ given by a circuit of size $s$ factors as $g^eh$, where $g$ and $h$ are coprime, then $g$ can be computed by a circuit of size $\text{poly}(e,\text{deg}(g),s)$. The question left open is to remove the dependency on $e$. In the special case where $f=g^e$, it was established that $g$ has circuit size $\text{poly}(\text{deg}(g),\text{size}(f))$. On the other hand, several algorithmic problems are NP-hard, eg.~computing the degree of the squarefree part, gcd, or lcm; even in the case of supersparse univariate polynomials \cite{plaisted1977sparse}.
   
Now, we discuss the closure results for classes more restrictive than VP (such as VF, VBP etc.). 
Unfortunately, Kaltofen's technique \cite{kaltofen1989factorization} for VF will give a superpolynomial-sized factor formula; as it heavily \emph{reuses} intermediate computations while working with linear algebra and Euclid gcd. The same holds for the class VBP. In contrast, extending the idea of \cite{dvir2009hardness}, Oliveira \cite{oliveira2016factors} showed that an $n$-variate polynomial with \emph{bounded individual degree} and computed by a formula of size $s$, has factors of formula size poly$(n,s)$. Furthermore, it was established that for a given $n$-variate individual-degree-$r$ polynomial, computed by a circuit (resp.~formula) of size $s$ and depth $\Delta$, there exists a $\text{poly}(n^r,s)$-time randomized algorithm that outputs any factor of $f$ computed by a circuit (resp.~formula) of depth $\Delta+5$ and size $\text{poly}(n^r,s)$.  We are not aware of any work specifically on VBP factoring, except a special case in \cite{kaltofen2008expressing}---it dealt with the elimination of a {\em single} division gate from skew circuits (also see Section \ref{ABP} \& Lemma \ref{div-elm})---and another special case result in \cite{jansen2011extracting} that was weakened later owing to proof errors.

Going beyond VP we can ask about the closure of VNP. B\"urgisser conjectured \cite[Conj.2.1]{burgisser2013completeness} that VNP is closed under factoring. Kaltofen's technique \cite{kaltofen1989factorization} for factoring VP circuits does not yield the closure of VNP and we are not aware of any further work on this. 

\smallskip
Recently, \emph{approximative} algebraic complexity classes like $\overline{\text{VP}}$ \cite{grochow2016boundaries} have become objects of interest, especially in the context of the geometric complexity program \cite{Mul12b, mul12, G15}. Interestingly, \cite[Thm.4.9]{mulmuley2017geometric} shows that the following three fundamental concepts are tightly related mainly due to circuit factoring results: {\bf 1)} efficient blackbox polynomial identity testing (PIT) for $\overline{\text{VP}}$, {\bf 2)} strong lower bounds against $\overline{\text{VP}}$, and {\bf 3)} efficiently computing an `explicit system of parameters' for the invariant ring of an explicit variety with a given group action.

$\overline{\text{VP}}$ contains families of polynomials of degree poly($n$) that can be approximated (infinitesimally closely) by poly($n$)-sized circuits. 
B\"urgisser \cite{burgisser2004complexity, burgisser2001complexity} discusses approximative complexity of factors, proving that low degree factors of high degree circuits have small approximative complexity. In particular, $\overline{\text{VP}}$ is closed under factoring \cite[Thm.4.1]{burgisser2001complexity}. Like the standard versions, closure of $\overline{\text{VF}}$ resp.~$\overline{\text{VBP}}$ is an open question. Recently, it has been shown that $\overline{\text{VF}}=$ width-$2$-$\overline{\text{VBP}}$ \cite{bringmann2017algebraic} while classically it is false \cite{allender2011power}. The new methods that we present extend nicely to approximative classes because of their analytic nature (Theorem \ref{thm-vf-bar}).

\smallskip
We conclude by stating a few reasons why closure results under factoring are interesting and non-trivial.
First, there are classes that are {\em not} closed under factors.
For example, the class of sparse polynomials; as a factor's sparsity may blowup super-polynomially \cite{von1985factoring}.
Closure under factoring indicates the robustness of an algebraic complexity class, as, it proves that all nonzero multiples of a \emph{hard} polynomial remain hard. For this reason, closure results are also important for proving lower bounds on the power of some algebraic proof systems \cite{forbes2016proof}.

Finally, factoring is the key reason why PIT, for VP, can be reduced to very special cases, and gets tightly related to circuit lower bound questions (like VP$\ne$VNP?). See \cite[Thm.4.1]{kabanets2003derandomizing} for whitebox PIT connection and \cite{AFGS17} for blackbox PIT. One of the central reasons is: Suppose a polynomial $f(\overline{y})$ is such that for a nonzero size-$s$ circuit $C$, $C(f(\overline{y}))=0$. Then, using factoring results for low degree $C$, one deduces that $f$ also has circuit size $\poly(s)$. This gives us the connection: {\em If we picked a ``hard'' polynomial $f$ then $f(\overline{y})$ would be a hitting-set generator (hsg) for $C$} \cite[Thm.7.7]{kabanets2003derandomizing}. Our work is strongly motivated by the open question of proving such a result for size-$s$ circuits $C$ that have high degree (i.e.~$s^{\omega(1)}$). Our first factoring result (Theorem \ref{thm1}) implies such a `hardness to hitting-set' connection for arbitrarily high degree circuits $C$ assuming that: the squarefree part $C_\text{sqfree}$ of $C$ has low degree. In such a case we only have to find a hitting-set for $C_\text{sqfree}$ which, as our result proves, has low algebraic circuit complexity. 


\vspace{-2mm}
\subsection{Our results }\label{sec-results}
\vspace{-1mm}

 Before stating the results, we describe some of the assumptions and notations used throughout the paper. Set $[n]$ refers to $\{1,2,\ldots,n\}$. Logarithms are wrt base $2$.
 
{\bf Field.} 
We denote the underlying field as $\F$ and assume that it is of characteristic $0$ and algebraically closed. For eg.~complex $\C$, algebraic numbers $\overline{\Q}$ or algebraic $p$-adics $\overline{\Q}_p$. All the results partially hold for other fields (such as $\R, \mathbb{Q}, \Q_p$ or finite fields of characteristic $>$degree of the input polynomial). For a brief discussion on this issue, see Section \ref{sec-extn}. 

{\bf Ideal.} We denote the variables $(x_1,\hdots,x_n)$ as $\overline{x}$. The {\em ideal} $I:= \langle \overline{x} \rangle$ of the polynomial ring will be of special interest, and its power ideal $I^d$, whose generators are all degree $d$ monomials in $n$ variables. Often we will reduce the polynomial ring modulo $I^d$ (inspired from {\em Taylor series of an analytic function around $\overline{0}$} \cite{taylor1715}).

{\bf Radical.} For a polynomial $f=\prod_i f_i^{e_i}$, with $f_i$'s coprime irreducible nonconstant polynomials and multiplicity $e_i>0$, we define the squarefree part as the {\em radical} {\em rad$(f):=\prod_i f_i$}. 

What can we say about these $f_i$'s if $f$ has a circuit of size $s$? Our main result gives a good circuit size bound when $\text{rad}(f)$ has small degree. A slightly more general formulation is:

\vspace{-1mm}
\begin{theorem}
\label{thm1}
If $f=u_0u_1$ in the polynomial ring $\F[\overline{x}]$, with  $\text{size}(f)+\text{size}(u_0) \leq s$, then {\em every} factor of $u_1$ has a circuit of size poly$(s+\text{deg}(\text{rad}(u_1)) )$. 
\end{theorem}

Note that Kaltofen's proof technique in the VP factoring paper \cite{kaltofen1989factorization} does not extend to the {\em exponential} degree regime (even when degree of rad$(f)$ is small) because it requires solving equations with deg$_{x_i}(f)$ many unknowns for some $x_i$, where deg$_{x_i}(f)$ denotes {\em individual degree} of $x_i$ in $f$, which can be very high. Also, basic operations like `determining the coefficient of a univariate monomial' become \#P-hard in the exponential-degree regime \cite{valiant1982reducibility}. The proof technique in Kaltofen's single factor Hensel lifting paper \cite[Thm.2]{kaltofen1987single} works only in the perfect-power case of $f=g^e$. It can be seen that rad$(f)$ ``almost'' equals $f/\gcd(f,\partial_{x_i}(f))$, but the gcd itself can be of exponential-degree and so one cannot hope to use \cite[Thm.4]{kaltofen1987single} to compute the gcd either. Univariate high-degree gcd computation is NP-hard \cite{plaisted1977new, plaisted1977sparse}.

Interestingly, our result when combined with \cite[Thm.3]{kaltofen1987single} implies that every factor $g$ of $f$ has a circuit of size polynomial in: $\siz(f)$, $\deg(g)$ and $\min\{\deg(\text{rad}(f)), \siz(\text{rad}(f)) \}$. We leave it as an open question whether the latter expression is polynomially related to $\siz(f)$. 

Theorem \ref{thm1} shows an interesting way to create {\em hard} polynomials. In the theorem statement let the size concluded be $(s+\text{deg}(\text{rad}(u_1)) )^e$, for some constant $e$. If one has a polynomial $f_1(x_1,\ldots,x_n)$ that is $2^{cn}$-hard, then any nonzero $f:= \prod_{i} f_i^{e_i}$ is also $2^{\Omega(n)}$-hard for {\em arbitrary} positive $e_i$'s, as long as $\sum_i\deg(f_i)\le 2^{\frac{cn}{e}-1}$.

\smallskip
In general, for a high degree circuit $f$, rad$(f)$ can be of high degree (exponential in size of the circuit). Ideally, we would like to show that every  degree $d$ factor of $f$ has poly$(\siz(f),d)$-size circuit. The next theorem reduces the above question to a special kind of modular division, where the denominator polynomial may {\em not} be invertible but the quotient is well-defined (eg.~$x^2/x \mod x$). All that remains is to somehow eliminate this kind of {\em non-unit division} operator (which we leave as an open question).

\vspace{-0.7mm}
\begin{theorem}
\label{thm2}
If $f \in \F[\overline{x}]$ can be computed by a circuit of size $s$, then any degree $d$ factor of $f$ is of the form $A/B \bmod \langle \overline{x} \rangle^{d+1}$ where polynomials $A, B$ have circuits of size $\poly(sd)$.
\end{theorem}

Note that in Theorem \ref{thm2}, $B$ may be non-invertible in $\F[\overline{x}]/\langle \overline{x} \rangle^{d+1}$ and may have a high degree (eg.~$2^s$). So, we cannot use the famous trick of Strassen to do division elimination here \cite{strassen1973vermeidung}. 

\smallskip
We prove uniform closure results, under factoring, for the algebraic complexity classes defined below. Let $s : \mathbb{N} \longrightarrow \mathbb{N}$ be a function. Define the class VF$(s)$ to contain families $\{f_n\}_n$ such that $n$-variate $f_n$  can be computed by an algebraic formula of size $\poly(s(n))$ and has degree $\poly(n)$. Similarly, $\text{VBP}(s)$ contains families $\{f_n\}_n$ such that $f_n$ can be computed by an ABP of size $\poly(s(n))$ and has degree $\poly(n)$. Finally, VNP$(s)$ denotes the class of families $\{f_n\}_n$ such that $f_n $ has witness size $\poly(s(n))$, verifier circuit size $\poly(s(n))$, and has degree $\poly(n)$.

\vspace{-0.7mm}
\begin{theorem}
\label{thm3}
The classes $\text{VF} (n^{\log n}),\text{VBP} (n^{\log n}),\text{VNP} (n^{\log n})$ are all closed under factoring. 

Moreover, there exists a randomized $\poly(n^{\log n})$-time algorithm that: for a given $n^{O(\log n)}$ sized formula (resp.\ ABP) $f$ of $\poly(n)$-degree, outputs $n^{O(\log n)}$ sized formula (resp.\ ABP) of a nontrivial factor of $f$ (if one exists). 
\end{theorem}
\vspace{-1mm}
\noindent {\bf Remark.} The ``time-complexity'' in the algorithmic part makes sense only in certain cases. For example, when $\F\in\{\Q, \Q_p, \F_q\}$, or when one allows computation in the BSS-model \cite{blum1989theory}. In the former case our algorithm takes $\poly(n^{\log n})$ bit operations (assuming that the characteristic is zero or larger than the degree; see Theorem \ref{thm-betterThm3} in Section \ref{sec-not-clos}).

\smallskip
It is important to note that Theorem \ref{thm3} does not follow by invoking Kaltofen circuit factoring \cite{kaltofen1989factorization} and VSBR transformation \cite{valiant1983fast} from circuit to log-depth formula. Formally, if we are given a formula (resp.\ ABP) of size $n^{O(\log n)}$ and degree $\poly(n)$, then it has factors which can be computed by a circuit of size $n^{O(\log n)}$ and depth $O(\log n)$. If one converts the factor circuit to a formula (resp.\ ABP), one would get the size upper bound of the factor formula to be a much larger $( n^{O(\log n)})^{\log n}= n^{O(\log^2 n)}$. Moreover, Kaltofen's methods crucially rely on the circuit representation to do linear algebra, division with remainder, and Euclid gcd in an efficient way; a nice overview of the implementation level details to keep in mind is \cite[Sec.3]{kopparty2015equivalence}.

Our proof methods extend to the approximative versions $\mathcal{C}(n^{\log n})$ for $\mathcal{C}\in$ $\{\overline{\text{VF}}, \overline{\text{VBP}}, \overline{\text{VNP}} \}$ as well (Theorem \ref{thm-vf-bar}). 

As before, Theorem \ref{thm3} has an interesting lower bound consequence: If $f$ has VF (resp.~VBP resp.~VNP) complexity $n^{\omega(\log n)}$ then any nonzero $fg$ has similar hardness (for $\deg(g)\le\poly(n)$). 

In fact, the method of Theorem \ref{thm3} yields a formula factor of size $s^e d^{2\log d}$ for a given degree-$d$ size-$s$ formula ($e$ is a constant). This means--- If determinant $\text{det}_n$ requires $n^{a\log n}$ size formula, for $a>2$, then {\em any} nonzero degree-$O(n)$ multiple of $\text{det}_n$ requires $n^{\Omega(\log n)}$ size formula. 

Similarly, if we conjecture that a VP-complete polynomial $f_n$ (say the homomorphism polynomial in \cite[Thm.19]{DurandMMRS14}) has $n^{a\log n}$ ABP complexity, for $a>4$, then {\em any} nonzero degree-$O(n)$ multiple of $f_n$ has $n^{\Omega(\log n)}$ ABP complexity.

\vspace{-2mm}
\subsection{Proof techniques} \label{sec-tech}
\vspace{-1mm}

We begin by describing the new techniques that we have developed. Since they also give a new viewpoint on classic properties, they may be of independent interest. The techniques are {\em analytic} at heart (\cite{book-KP12} has a good historical perspective). The way they appear in algebra is through the {\em formal power series} ring $\F[[x_1,\ldots,x_n]]$. The elements of this ring are multivariate formal power series, with degree as precision. So, an element $f$ is written as $f= \sum_{i=0}^{\infty} f^{=i}$, where $f^{=i}$ is the {\em homogeneous part} of degree $i$ of $f$. In algebra texts it is also called the {\em completion} of $\F[x_1,\ldots,x_n]$ wrt the ideal $\langle x_1,\ldots,x_n\rangle$ (see \cite[Chap.13]{kemper2010course}). The {\em truncation} $f^{\le d}$, i.e.~homogeneous parts up to degree $d$, can be obtained by reducing modulo the ideal $\langle \overline{x}\rangle^{d+1}$. Here $d$ is seen as the {\em precision} parameter of the respective approximation of $f$.

The advantages of the ring $\F[[\overline{x}]]$ are many. They usually emerge because of the {\em inverse} identity $(1-x_1)^{-1} = \sum_{i\ge0}x_1^i$ , which would not have made sense in $\F[\overline{x}]$ but is available now. 
First, we introduce a factorization pattern of a polynomial $f$, over the power series ring, under a random linear transformation. Next, we discuss how this factorization helps us to bound the size of factors of the original polynomial.

\medskip\noindent
\textbf{Power series complete split:}
We are interested in the {\em complete} factorization pattern of a polynomial $f(x_1,\hdots,x_n)$.
We can view $f$ as a univariate polynomial in one variable, say $x_n$, with coefficients coming from $\F[x_1,\hdots,x_{n-1}]$. It is easy to connect linear factors with the roots: $x_n-g$ is a factor of $f$ iff $f(x_1,\hdots,x_{n-1},g(x_1,\hdots,x_{n-1}))=0$. 

Of course, one should not expect that a polynomial always has a factor which is linear in one variable. But, if one works with an algebraically closed field, then a univariate polynomial completely splits into linear factors (also see the {\em fundamental theorem of algebra} \cite[\S 2.5.4]{courant1996mathematics}). So, if we go to the algebraic closure of $\F(x_1,\hdots,x_{n-1})$, any multivariate polynomial which is monic in $x_n$ will split into factors all linear in $x_n$. A representation of the elements of $\overline{\F(x_1,\hdots,x_{n-1}) }$ as a finite circuit is impossible (eg.~$\sqrt{x_1}$). On the other hand, we show in this work that {\em all} the roots (wrt a new variable $y$) are actually elements from $\F[[x_1,\hdots,x_{n}]]$, after a \emph{random} linear transformation on the variables, $\tau: \overline{x}\mapsto \overline{x}+ \overline{\alpha}y + \overline{\beta}$, is applied (Theorem \ref{thm-complete-split}). Note-- By a random choice $\alpha\in_r\F$ we will mean that choose randomly from a fixed finite set $S\subseteq \F$ of appropriate size (namely $>\deg(f)$). This will be in the spirit of \cite{Sch80}.

Our proof of the existence of power series roots is {\em constructive}, as it also gives an algorithm to find approximation of the roots up to any precision, using  formal power series version of the Newton iteration method (see \cite[Thm.2.31]{burgisser2013algebraic}). 
We try to explain the above idea using the following example. Consider $f=(y^2-x^3)\in \F[x,y]$. Does it have a factor of the form $y-g$ where $g \in \F[[x]]$ ? The answer is clearly `no' as $x^{3/2}$ does not have any power series representation in $\F[[x]]$. But, what if we shift $x$ randomly? For example, if we use the shift $y \mapsto y, x \mapsto x+1$. Then, by Taylor series around $1$, we see that $(x+1)^{3/2}$ has a power series expansion, namely $1+\frac{3}{2}x+\frac{3/2 \times 1/2}{2!} x^2+ \hdots$.

Formally, Theorem \ref{thm-complete-split} shows that under a random $\tau:\overline{x} \mapsto \overline{x}+ \overline{\alpha}y + \overline{\beta}$ where $\overline{\alpha}, \overline{\beta} \in_r \F^n$, polynomial $f$ can be factored as  $f(\tau\overline{x})= \prod_{i=1}^{d_0}(y-g_{i})^{\gamma_{i}}$, where $g_{i} \in  \F[[\overline{x}]]$ with the constant terms $g_{i}(\overline{0})$ being distinct, $d_0:= \deg(\text{rad}(f))$ and $\gamma_{i}>0$. 


\medskip\noindent
\textbf{Reducing factoring to computing power series root approximations:}
Using the split Theorem \ref{thm-complete-split},  we show that multivariate polynomial factoring reduces to power series root finding up to certain precision. Following the above notation $f$ splits as $f(\tau\overline{x})= \prod_{i=1}^{d_0}(y-g_{i})^{\gamma_{i}}$. For all $t\geq 0$, it is easy to see that
$f(\tau\overline{x}) \,\equiv\, \prod_{i=1}^{d_0}(y-g_{i}^{\le t})^{\gamma_{i}} \mod I^{t+1}$, where $I:= \langle x_1,\hdots,x_{n} \rangle$. Note that there is a one-one correspondence, induced by $\tau$, between the polynomial factors of $f$ and $f(\tau\overline{x})$ ($\because \tau$ is invertible and $f$ is $y$-free). We remark that the leading-coefficient of $f(\tau\overline{x})$ wrt $y$ is a nonzero element in $\F$; so, we call it {\em monic} (Lemma \ref{lem-monic}). Next, we show case by case how to find a {\em polynomial} factor of $f(\tau\overline{x})$ from the approximate power series roots.

\smallskip \noindent
{\em Case 1- Computing a linear factor of the form $y-g(\overline{x})$:}
If the degree of the input polynomial is $d$, all the non-trivial factors have degree $\leq (d-1)$. So, if we compute the approximations of all the power series roots (wrt $y$) up to precision of degree $t=d-1$, then we can recover all the factors of the form  $y-g(x_1,\ldots,x_{n})$. Technically, this is supported by the uniqueness of the power series factorization (Proposition \ref{prop-ufd}).

\smallskip\noindent
{\em Case 2- Computing a monic non-linear factor:} Assume that a factor $g$ of total degree $t$ is of the form $y^k+c_{k-1}y^{k-1}+...+c_1y+c_0$, where for all $i$, $c_i\in \F[\overline{x}]$. Now this factor $g$ also splits into linear (in $y$) factors above $\F[[\overline{x}]]$ and obviously these linear factors are also linear factors of the original polynomial $f(\tau\overline{x})$. So we have to take the right combination of some $k$ power series roots, with their approximations (up to the degree $t$ wrt $\overline{x}$), and take the product mod  $I^{t+1}$. Note that if we only want to give an existential proof of the size bound of the factors, we need not find the combination of the power series roots forming a factor algorithmically. 
 Doing it through brute-force search takes exponential time (${d\choose k}$ choices). Interestingly, using a classical (linear algebra) idea due to Kaltofen, it can be done in randomized polynomial time. We will spell out the ideas later, while discussing the algorithm part of Theorem \ref{thm3}.
 
\medskip
Once we are convinced that looking at approximate (power series) roots is enough, we need to investigate methods to compute them. We will now sketch two methods. The first one approximates all the roots {\em simultaneously} up to precision $\delta$. The next ones approximate the roots {\em one at a time}. In the latter, multiplicity of the root plays an important role.
 
\smallskip \noindent
\textbf{Recursive root finding (allRootsNI):} We \emph{simultaneously} find the approximations of all the power series roots $g_{i}$ of $f(\tau\overline{x})$. At each recursive step we get a better precision wrt degree. We show that knowing approximations $g_{i}^{<\delta}$, of $g_{i}$ up to degree $\delta-1$, is enough to (simultaneously for all $i$) calculate approximations of $g_{i}$ up to degree $\delta$. This new technique, of finding approximations of the power series roots, is at the core of Theorem \ref{thm1}. 

First, let us introduce a nice identity. From now on we assume $f(\overline{x},y)= \prod_i(y-g_{i})^{\gamma_{i}}$ (i.e.~relabel $f(\tau\overline{x})$). By applying the derivative operator $\partial_{y}$, we get a classic identity (which we call {\em logarithmic derivative identity}):
$(\partial_{y} f)/ f=$ $\sum_{i} \gamma_{i}/(y-g_{i})$ .
Reduce the above identity modulo $I^{\delta+1}$ and define $\mu_{i}:= g_{i}(\overline{0}) \equiv g_i \mod I$. This gives us (see Claim \ref{lem-logDer-Id}):
$$\frac{\partial_{y} f}{ f} \,=\, \sum_{i=1}^{d_0} \frac{\gamma_{i}}{y-g_{i}} \,\,\equiv\,\, \sum_{i=1}^{d_0}\frac{\gamma_{i}}{y-g_{i}^{<\delta}} \,+\, \sum_{i=1}^{d_0}\frac{\gamma_{i}\cdot g_{i}^{=\delta}}{(y-\mu_{i})^2} \,\;\bmod {I^{\delta+1}} .$$  

In terms of the $d_0$ unknowns $g_{i}^{=\delta}$, the above is a linear equation. (Note- We treat $\gamma_i, \mu_i$'s as known.) 
As $y$ is a free variable above, we can fix it to $d_0$ ``random'' elements $c_i$ in $\F$, $i\in[d_0]$. One would expect these fixings to give a linear system with a unique solution for the unknowns. We can express the system of linear equations succinctly in the following matrix representation:  
$M\cdot v_\delta \,=\, W_\delta \,\bmod I^{\delta+1}$. 
Here $M$ is a $d_0 \times d_0$ matrix; each entry is denoted by $M(i,j):=\frac{\gamma_{i}}{(c_i-\mu_{j})^2}$. Vector $v_\delta$ resp.~$W_\delta$ is a $d_0 \times 1$ matrix where each entry is denoted by $v_{\delta}(i):=g_{i}^{=\delta}$ resp.~$W_\delta(i):= \frac{\partial_{y}f}{f}\big\rvert_{y=c_i} - G_{i,\delta}$, where $ G_{i,\delta}:= \sum_{k=1}^{d_0} \gamma_{k}/(c_i-g_{k}^{<\delta})$ . We ensure that $\{c_i, \mu_{i} \mid i\in[d_0]\}$ are distinct, and show that the determinant of $M$ is non-zero (Lemma \ref{lem-inv-det}). So, by knowing approximations up to $\delta-1$, we can recover $\delta$-th part by solving the above system as $v_\delta \,=\, M^{-1}W_\delta \,\bmod I^{\delta+1}$. An important point is that the random $c_i$'s will ensure: all the reciprocals involved in the calculation above do exist mod $I^{\delta+1}$.

{\em Self-correction property:} Does the above recursive step need an exact $g_{i}^{<\delta}$? We show the self correcting behavior of this process of root finding, i.e.~in this iterative process there is no need to filter out the ``garbage'' terms of degree $\ge\delta$ in each step. If one has recovered $g_{i}$ correct up to degree $\delta-1$, i.e.~say we have calculated $g_{i,\delta-1}' \in \F(\overline{x})$ such that $g_{i,\delta-1}'\equiv g_{i}^{< \delta} \bmod  I^{\delta}$, and say we solve $M \tilde{v}_{\delta}= \widetilde{W}_{\delta}$ exactly, where  $\widetilde{W}_\delta(i):= \frac{\partial_{y} f}{f}\big\rvert_{y=c_i} -\widetilde{G}_{i,\delta}$, and $ \widetilde{G}_{i,\delta}:=\sum_{k=1}^{d_0} \gamma_{k}/(c_i-g_{k,\delta-1}')$. Still, we can show that $g_{i,\delta}':=$ $g_{i,\delta-1}'+ \tilde{v}_{\delta}(i)$ $\equiv g_{i}^{\leq \delta} \bmod I^{\delta+1}$ (Claim \ref{clm-selfCorr}). So, we made progress in terms of the precision (wrt degree).

\medskip\noindent
\textbf{Rapid Newton Iteration with multiplicity:} We show that from allRootsNI, we can derive a formula that finds $g_{1}^{<2^{t+1}}$ using {\em only} $g_{1}^{<2^{t}}$, i.e.~the process has quadratic convergence and it does not involve roots other than $g_1$. Rewrite $\partial_{y} f/f= \sum_{i=1}^{d_0} \gamma_{i}/(y-g_{i})= (1+ L_1)\gamma_{1}/(y-g_{1})$, where $L_1:= \sum_{1<i\le d_0} \frac{\gamma_{i}}{y-g_{i}} \cdot \frac{y-g_{1}}{\gamma_{1}}$. This implies $f/\partial_{y}f=$ $(1+L_{1})^{-1} \cdot (y-g_{1})/\gamma_{1}$. 
Now, if we put $y= y_t:= g_{1}^{< 2^t}$, then $y_t-g_i= g_{1}^{< 2^t}-g_{i}$ is a unit in $\F[[\overline{x}]]$ for $i\ne1$ ($\because$ it is a nonzero constant mod $I$). Also, $y_t-g_1=$ $g_{1}^{<2^t}-g_{1} \equiv 0 \mod I^{2^{t}}$, implying $L_1\rvert_{y=y_t}\equiv 0 \mod I^{2^{t}}$. Thus, $\left(L_{1}\cdot(y-g_1)\right)\big\rvert_{y= y_t} \equiv 0 \mod I^{2^{t+1}}$.  

Hence, $f/\partial_{y}f\big \rvert_{y=y_t} =$ $(y_t- g_1)/\gamma_{1} \mod I^{2^{t+1}}$.

This shows that, if $f(\overline{x},y)=(y-g)^eh$, where $h\vert_{y=g}\ne0 \mod I$ and $e>0$, then the power series for $g$ can be approximated by the recurrence:
$$y_{t+1} \,:=\, y_t \,-\, e\cdot \frac{f}{\partial_y f}\bigg\rvert_{y=y_t} $$ where $y_t \equiv g \mod I^{2^t}$. This we call a {\em generalized Newton Iteration} formula, as it works with any multiplicity $e>0$. In fact, when $e=1$, $g$ is called a {\em simple root} of $f$; the above is an alternate proof of the classical Newton Iteration (NI) \cite{newton1669} that finds a simple root in a recursive way (see Lemma \ref{lem-NI}). It is well known that NI fails to approximate the roots that repeat (see \cite{lecerf2002quadratic}). In that case  either NI is used on the function $f/\partial_{y}f$ or, though less frequently, the generalized NI is used in numerical methods (see \cite[Eqn.6.3.13]{dahlquist2008numerical}).

There is a technical point about our formula for $e\ge2$. The denominator $\partial_{y}f\rvert_{y=y_t}$ is zero mod $I$, thus, its reciprocal does not exist! However, the ratio $(f/\partial_{y}f) \big \rvert_{y=y_t}$ does exist in $\F[[\overline{x}]]$. On the other hand, if $e=1$ then the denominator $\partial_{y}f\rvert_{y=y_t}$ is nonzero mod $I$, thus, it is invertible in $\F[[\overline{x}]]$ and that allows fast algebraic circuit computations ({\em classical NI}).

\medskip
We can compare the NI formula with the recurrence formula (which we call {\em slow} Newton Iteration) used in \cite[Eqn.5]{dvir2009hardness}, \cite[Lem.4.1]{oliveira2016factors} for root finding. The slow NI formula is $Y_{t+1}= Y_t - \frac{f(\overline{x},Y_t)}{\partial_y f(\overline{0},Y_1)}$, where  $Y_t  \equiv g \mod I^{t}$. The rate of convergence of this iteration is linear, as it takes $\delta$ many steps (instead of $\log\delta$) to get precision up to degree $\delta$.
One can also compare NI with other widespread processes like multifactor Hensel lifting \cite[Sec.15.5]{von2013modern}, \cite{zassenhaus1969hensel} and the implicit function theorem paradigm \cite[Sec.1.3]{book-KP12}, \cite{KS16, PSS16}; however, we would not like to digress too much here as the latter concept covers a whole lot of ground in mathematics.

\vspace{-2mm}
\subsection{Proof overview}
\vspace{-1mm}

In all our proofs, we use the reduction of factoring to power series root approximation, and then find the latter using various techniques described before.

\medskip\noindent
\textbf{Proof idea of Theorem \ref{thm1}:} We use the technique of allRootsNI to find the approximations of all the power series roots of $f(\tau\overline{x})$. As we already discussed how to find a polynomial factor $g$ of $u_1$ (that divides $f$) from the roots of $f(\tau\overline{x})$, what remains is to analyze the size bound for power series roots that we get from allRootsNI process. We note a few crucial points that help to prove the size bound. 

Let $d_0$ be the degree of $\rad(u_1)$. The number of distinct power series roots, of $u_1(\tau\overline{x})$ wrt $y$, is $d_0$. It suffices to approximate the power series roots up to degree $d_0$, as any nontrivial polynomial factor of $\rad(u_1(\tau\overline{x}))$ has degree less than $d_0$. Also, a size bound on these factors of the radical directly gives a size bound on the polynomial factor $g$.

The logarithmic derivative satisfies: $\partial_y \log f(\tau\overline{x}) =$ $\partial_y \log u_0(\tau\overline{x})+$ $\partial_y \log u_1(\tau\overline{x})$.
Since we have size $s$ circuits for both $f$ and $u_0$, and $y$ is later fixed to random $c_i$'s in $\F$, we can approximate the first two logarithmic derivative circuits modulo $I^{d_0+1}$. This approximates $\partial_y u_1(\tau\overline{x})/u_1(\tau\overline{x})$.

On this, allRootsNI process is used to approximate the power series roots of $u_1(\tau\overline{x})$ up to degree $d_0$. The self correcting behavior of the allRootsNI is crucial in the size analysis. If one had to truncate modulo $I^{d_0+1}$ at each recursive step, there would have been a multiplicative blowup (by $d_0$) in each step, which would end up with an exponential blow up in the size of the roots. The self correcting property allows to complete allRootsNI process, with division gates and partially correct roots $g_{i,\delta}'$, to get a circuit of size $\poly(sd_0)$. The truncation modulo $I^{d_0+1}$, to get a root of degree $\le d_0$, is performed only once in the end. See Section \ref{sec-pf-thm1}.

The steps in the proof of Theorem \ref{thm1} are constructive. However, to claim that we have an efficient algorithm we will need, in advance, the multiplicity of each of the $d_0$ roots. It is not clear how to find them efficiently, even in the univariate case $n=1$, as the multiplicity could be exponentially large.

\medskip\noindent
\textbf{Proof idea of Theorem \ref{thm2}:} The main technique used is NI with multiplicity. The main barrier in resolving high degree case is handling roots with high multiplicities (i.e. super-polynomial in size $s$). If all the roots of the polynomial have multiplicity equal to one, then we can use classical Newton iteration. If the multiplicity of a root is low (up to poly($s$)), we can differentiate and bring down the multiplicity to one. In Theorem \ref{thm1}, we handled the case of high multiplicity by assuming that the radical has small degree. 

So, the only remaining case is when both the number of roots, and their  multiplicities, are high. Newton iteration with multiplicity helps here. Note that we need to know the multiplicity of the root \emph{exactly} to apply NI with multiplicity; here, we will simply guess them non-uniformly. In the end, the process gives a circuit of size $\poly(sd)$ with division gates, giving the root mod $I^{d+1}$. By using a standard method the division gates can all be pushed ``out'' to the root. See Section \ref{sec-pf-thm2}.

\medskip\noindent
\textbf{Proof idea of Theorem \ref{thm3}:} Here, we show the closure under factoring for the algebraic complexity classes $VF(n^{\log n}), VBP(n^{\log n}), VNP(n^{\log n})$. In fact, we also give randomized $n^{O(\log n)}$-time algorithm to output the factors as formula (resp.~algebraic branching program). The key technique here is the classical Newton Iteration. 
The crucial advantage of NI over other approaches of power series root finding is that NI requires only $\log d$ steps to get precision up to degree $d$, whereas allRootsNI, \cite[Eqn.5]{dvir2009hardness} or \cite[Lem.4.1]{oliveira2016factors} require $d$ steps. This leads to a slower size blow up in the case of restricted models like formula or ABP. 

In a formula resp.~ABP, we cannot reuse intermediate computations. So each recursive step of NI incurs a blow up by $d^2$, as one needs to substitute $y_t$ in a degree $d$ polynomial $f(y)$ which may require that many copies of $y_t$-powers. But, as the NI process has only $\log d$ steps, ultimately, we get $d^{2\log d}$ blow up in the size bound. This is the main idea of the existential results in Theorem \ref{thm3}. Moreover, an interesting by-product is that VF, VBP and VNP are closed under factors if we only consider polynomials with individual degree {\em constant} (also see \cite{oliveira2016factors}).

All the steps in the proof of the existential result are algorithmically efficient except for one. We are recovering all the power series roots and multiplying a few of them to get a non-trivial factor. How do we choose the right combination of the roots which gives a non-trivial factor? If we search for the right combination in a brute-force way, it would need exponential (like $2^d$) time complexity. Here, linear algebra saves us; the idea dates back to Kaltofen's algorithm for bivariate factoring. Our contribution lies in the careful analysis of the different steps,  coming up with a new algorithm for computing gcd, and making sure that everything  works with formulas resp.~ABPs.  

Consider the transformed polynomial $f(\tau\overline{x})$ that is monic and degree $d$ in $y$. It will help us if we think of this polynomial as a bivariate (i.e.~in $y$ and a new degree-counter $T$). This somewhat reduces the problem to a two-dimensional case and makes the modular computations feasible (see \cite[Sec.1.2.2]{kopparty2015equivalence}). So, we need to apply the map $\overline{x} \mapsto T\overline{x}$, where $T$ is a new formal variable; call the resulting polynomial $\tilde{f}(\overline{x},T,y)$. This map preserves the power series roots; in fact, we can get the roots of $f(\tau\overline{x})$ by putting $T=1$. Now comes the most important idea in the algorithm.
Approximate a root $g_i$ up to large enough precision (say $k:=2d^2$). Solve the system of linear equations $u = (y \,-\, g_i^{\leq k}(Tx))\cdot v \bmod T^{k+1}$ for monic polynomials $u, v$. Then, $u$ will give a non-trivial factor when we compute $\gcd_y(u,\tilde{f})$. Intuitively, the gcd gives us the irreducible polynomial factor whose root is the power series $g_i$ that we had earlier computed by NI.

Note that a modified gcd computation is needed to actually get a factor as a formula resp.~ABP. If one uses the classical Euclidean algorithm, there are $d$ recursive steps to execute; at each step there would be a blow up of $d$ (as for formula or ABP, we cannot reuse any intermediate computation). So, in this approach (eg.~the one used in \cite{kopparty2015equivalence}), gcd of the two formulas will be of exponential size. The way we achieve a better bound is by first using NI to approximate {\em all} the power series roots of $u$ and $\tilde{f}$. Subsequently, we filter the ones that appear in both to learn the gcd. There is an alternate way as well based on our Claim \ref{clm-irred-fac}. See Section \ref{sec-pf-thm3}.


\section{Preliminaries}
\label{sec-pre}
\vspace{-1mm}

In our proofs we will need some basic results about formulas, ABPs and circuits. 
In particular, we can efficiently eliminate a division gate, we can extract a homogeneous part, and we can compute a (first-order) derivative. Also, see \cite[Sec.2]{kopparty2015equivalence}.

Determinant is in VBP and is computable by a $n^{O(\log n)}$ size formula.

We will use properties of gcd($f,g$) and a related determinant polynomial called {\em resultant}. 

To save space we have moved the well known details to Section \ref{prel}.

\section{Power series factorization of polynomials}
\label{sec-split}
\vspace{-1mm}

  Instead of looking into the factorization over $\F[\overline{x}]$, we look into the more analytic factorization pattern of a polynomial over $\F[[x_1,\hdots,x_n]]$, namely, formal power series of $n$-variables over field $\F$. To talk about factorization, we need the notion of \emph{uniqueness} which the following proposition ensures.

\begin{proposition}\cite[Chap.VII]{zariski1975commutative}\label{prop-ufd}
Power series ring $\F[[x_1,\ldots,x_n]]$ is a unique factorization domain (UFD), and so is $\F[[\overline{x}]][y]$.
\end{proposition}

As discussed before, we need to first apply a random linear map, that will make sure that the resulting polynomial splits completely over the ring $\F[[\overline{x}]]$. (Recall: $\F$ is algebraically closed.)

\begin{theorem}[Power Series Complete Split]
\label{thm-complete-split}
Let $f \in \F[\overline{x}] $ with deg$(\text{rad}(f))=: d_0>0$. Consider $\alpha_i, \beta_i \in_{r} \mathbb{F} $ and the map $\tau : x_i\mapsto \alpha_{i}y+x_i + \beta_i$, $i\in[n]$, where $y$ is a new variable.

Then, over $\F[[\overline{x}]]$, $f(\tau \overline{x})= k\cdot \prod_{i\in[d_0]} (y-g_i)^{\gamma_i}$, where $k\in\F^*$, $\gamma_i > 0$, and $g_i(\overline{0}):= \mu_i$. Moreover, $\mu_i$'s are distinct nonzero field elements.
\end{theorem}

\begin{proof}
Let the irreducible factorization of  $f$ be $\prod_{i\in[m]} f_i^{e_i}$. We apply a random $\tau$ so that $f$, thus all its factors, become monic in $y$ (Lemma \ref{lem-monic}). 
The monic factors $\tilde{f_i}:= f_i(\tau \overline{x})$ remain irreducible ($\because$ $\tau$ is invertible). Also, $\tilde{f_i}(\overline{0},y)= f_i(\overline{\alpha}y+ \overline{\beta})$ and $\partial_y\tilde{f_i}(\overline{0},y)$ remain coprime ($\because \overline{\beta}$ is random, apply Lemma \ref{lem-coprime}). In other words, $\tilde{f_i}(\overline{0},y)$ is square free (Lemma \ref{lem-sq-free}).

In particular, one can write $\tilde{f_1}(\overline{0},y)$ as $\prod_{i=1}^{\deg(f_1)} (y-\mu_{1,i})$ for distinct nonzero field elements $\mu_{1,i}$ (ignoring the constant which is the coefficient of the highest degree of $y$ in $\tilde{f_1}$). Using classical Newton Iteration (see Lemma \ref{lem-NI} or \cite[Thm.2.31]{burgisser2013algebraic}), one can write $\tilde{f_1}(\overline{x},y)$ as a product of power series $\prod_{i=1}^{\text{deg}(f_1)} (y-g_{1,i})$, with $g_{1,i}(\overline{0}):= \mu_{1,i}$. Thus, each $f_i(\tau \overline{x})$ can be factored into linear factors in $\F[[\overline{x}]][y]$.

As $f_i$'s are irreducible coprime polynomials, by Lemma \ref{lem-coprime}, it is clear that $\tilde{f_i}(\overline{0},y)$, $i\in[m]$, are mutually coprime. In other words, $\mu_{j,i}$ are distinct and they are $\sum_i\deg(f_i)= d_0$ many. Hence, $f(\tau \overline{x})$ can be completely factored as $\prod_{i\in[m]} f_i(\tau \overline{x})^{e_i}=$ $\prod_{i\in[d_0]} (y-g_i)^{\gamma_i}$, with $\gamma_i>0$ and the field constants $g_i(\overline{0})$ being distinct.
\end{proof}

\begin{corollary} \label{cor1}
Suppose $g$ is a polynomial factor of $f$. As before let $f(\tau \overline{x})=\prod_{i\in[m]} f_i(\tau \overline{x})^{e_i}$ $= k\cdot \prod_{i\in[d_0]} (y-g_i)^{\gamma_i}$. As $g(\tau x) \mid f(\tau \overline{x})$ we deduce that $g(\tau x)= k'\prod (y-g_i)^{c_i}$ with $0\le c_i \leq \gamma_i$. 

Moreover, we can get back $g$ by applying $\tau^{-1}$ on the resulting polynomial $g(\tau\overline{x})$.
\end{corollary}

\section{Main Results} \label{sec-main}
\vspace{-1mm}

This section proves Theorems \ref{thm1}--\ref{thm3}. The proofs are self contained and we assume for the sake of simplicity that the underlying field $\F$ is algebraically closed and has characteristic $0$. When this is not the case, we discuss the corresponding theorems in Section \ref{sec-extn}.  

\subsection{Factors of a circuit with low-degree radical: Proof of Theorem \ref{thm1}}
\label{sec-pf-thm1}
\vspace{-1mm}

In this section, we use Theorem \ref{thm-complete-split} and allRootsNI to partially solve the case of circuits with exponential degree (stated in \cite{Kaltofen86} and studied in \cite{kaltofen1987single, burgisser2004complexity}).

\begin{proof}[Proof of Theorem \ref{thm1}]

From the hypothesis $f=u_0u_1$. Define $\deg(f)=:d$. Suppose $u_1=h_1^{e_1}\hdots h_m^{e_m} $, where $h_i$'s are coprime irreducible polynomials. Let $d_0$ be the degree of $\rad(u_1) = \prod_i h_i$. Note that deg($h_i), m \leq d_0$ and the multiplicity $e_i\le d\le s^{O(s)}$, where $s$ is the size bound of the input circuit. Thus, to get the size bound of any factor of $u_1$, it is enough to show that for each $i$, $h_i$ has a circuit of size poly$(sd_0)$.

Using Theorem \ref{thm-complete-split}, we have $ \tilde{f}(\overline{x},y) := f(\tau \overline{x})= k\cdot u_0(\tau \overline{x})\cdot \prod_{i\in[d_0]} (y-g_i)^{\gamma_i}$, with $ g_{i}(\overline{0}) :=\mu_i$ being distinct. From Corollary \ref{cor1} we deduce that $h_i(\tau \overline{x})= k_i \prod_{i\in[d_0]} (y-g_i^{\leq d_0})^{\delta_i} \bmod I^{d_0+1} $, with ideal $I:= \langle x_1, \hdots, x_n \rangle$, exponent $ \delta_i \in \{0,1\} $ and  nonzero $k_i \in \F$. We can get $h_i$ by applying $\tau^{-1}$. Hence, it is enough to bound the size of $g_i^{\leq d_0}$.

Let $\tilde{u_0} := u_0(\tau \overline{x})$. From the repeated applications of Leibniz rule of the derivative $\partial_y$, we deduce,
$\partial_{y} \tilde{f}/\tilde{f} \,=\, \partial_{y}\tilde{u_0}/\tilde{u_0} + \sum_{i=1}^{d_0} \gamma_{i}/(y-g_{i})$. (Recall: $\partial_y (FG) = (\partial_y F)G+ F(\partial_y G)$.)

At this point we move to the formal power series, so that the reciprocals can be approximated as polynomials.
Note that $y-g_{i}$ is invertible in $\mathbb{F}[[\overline{x}]]$ when $y$ is assigned any value $c_i\in\F$ which is not equal to  $\mu_{i}$. We intend to find $g_{i} \bmod I^{\delta} $ inductively, for all $\delta\ge1$.  We assume that $\mu_{i}$'s and $\gamma_i$'s are known. Suppose, we have recovered up to $g_{i} \bmod I^{\delta} $ and we want to recover $g_{i} \bmod  I^{\delta+1} $. The relevant recurrence, for $\delta\ge1$, is: 

\begin{claim}[Recurrence]\label{lem-logDer-Id}
 $\sum_{i=1}^{d_0} \gamma_{i}\cdot g_{i}^{=\delta}/(y-\mu_{i})^2  \,\equiv\, \partial_{y} \tilde{f}/\tilde{f} \,-\, \partial_{y}\tilde{u_0}/\tilde{u_0} \,-\, \sum_{i} \gamma_{i}/(y-g_{i}^{<\delta}) \bmod I^{\delta+1}$. 
\end{claim}
\claimproof{lem-logDer-Id}{
Using a power series calculation (Lemma \ref{lem-series-id}), we have 
$\frac{1}{y-g_{i}} \equiv \frac{1}{y-\left(g_{i}^{<\delta}+g_{i}^{=\delta}\right)} \equiv \frac{1}{y-g_{i}^{<\delta}} + \frac{g_{i}^{=\delta}}{(y-\mu_{i})^2} \bmod I^{\delta+1} $. Multiplying by $\gamma_i$ and summing  over $i\in[d_0]$, the claim follows.
}

By knowing approximation up to the $\delta-1$ homogeneous parts of $g_i$, we want to find the $\delta$-th part by solving a linear system. For concreteness, assume that we have a rational function $g'_{i,\delta-1} := C_{i,\delta-1}/D_{i,\delta-1}$ such that $g'_{i,\delta-1} \equiv g_i^{<\delta} \bmod I^\delta$. Next, we show how to compute $g_i^{\le\delta}$. 

We recall the process as outlined in allRootsNI (Section \ref{sec-tech}). In the free variable $y$, we plug-in $d_0$ random field value $c_i$'s and get the following system of linear equations: $M\cdot v_{\delta} = W_{\delta}$, where $M$ is a $d_0 \times d_0$ matrix with $(i,j)$-th entry, $M(i,j):= \gamma_j/(c_{i}-\mu_{j})^2$. 
Column $v_{\delta}$ resp.~$W_{\delta}$ is a $d_0 \times 1$ matrix whose $i$-th entry is denoted $v_{\delta}(i)$ resp.~$(\partial_{y} \tilde{f}/\tilde{f} -\partial_{y} \tilde{u_0}/\tilde{u_0})\rvert_{y=c_{i}}$ $-\,\tilde{G}_{i,\delta}$, where $\tilde{G}_{i,\delta} := \sum_{j=1}^{d_0} \gamma_{j}/(c_{i}-g_{j,\delta-1}')$. Think of the solution $v_{\delta}$ as being both in $\F(\overline{x})^{d_0}$ and in $\F[[\overline{x}]]^{d_0}$; both the views help.

Now we will prove two interesting facts. First, $M$ is invertible (Lemma \ref{lem-inv-det}). Second, define $g'_{i,0}:=\mu_i$ and, for $\delta\ge1$, $g'_{i,\delta} :=$ $g'_{i,\delta-1}+v_{\delta}(i)$. Then, $g'_{i,\delta}$ approximates $g_i$ well:

\begin{claim}[Self-correction]
\label{clm-selfCorr}
Let $i\in[d_0]$ and $\delta\ge0$.
Then, $g'_{i,\delta}\equiv$ $g_i^{\le\delta} \bmod I^{\delta+1}$. 
\end{claim}

\claimproof{clm-selfCorr}{
We prove this by induction on $\delta$. It is true for $\delta=0$ by definition. Suppose it is true for $\delta-1$. 
This means we have $ g_{i,\delta-1}' \equiv g_{i}^{<\delta} \bmod I^{\delta} $ for all $i$. Let us write $g_{i,\delta-1}'=: g_{i}^{<\delta}+ A_{i,\delta} + A_{i,\delta}' $, where $A_{i,\delta}' \equiv 0 \bmod  I^{\delta+1}$ and $A_{i,\delta}$ is homogeneous of degree $\delta$. 
Hence, for $i\in[d_0]$, the linear constraint is: $\sum_{j=1}^{d_0} \gamma_{j}\cdot v_\delta(j)/(c_i-\mu_{j})^2 \,\equiv\, \partial_{y} \tilde{f}/\tilde{f} \,-\, \partial_{y} \tilde{u_0}/\tilde{u_0} \,-\,  \sum_{j} \gamma_{j}/(c_i-g_{j,\delta-1}')  \bmod  I^{\delta+1} $.

\smallskip
The ``garbage'' term $A_{j,\delta}$ in RHS can be isolated using Lemma \ref{lem-series-id} as: 
$1/(c_i-g_{j,\delta-1}') \equiv  \frac{1}{c_i-\left(g_{j}^{<\delta}+ A_{j,\delta}\right)} 
   \equiv 1/(c_i-g_j^{<\delta}) \,+\, A_{j,\delta}/(c_i-\mu_j)^2 \bmod I^{\delta+1} $.
So, we get: 
$$ \sum_{j=1}^{d_0} \frac{\gamma_{j}\cdot v_\delta(j)}{(c_i-\mu_j)^2} \,\equiv\,  \frac{\partial_{y} \tilde{f}}{\tilde{f}} \,-\, \frac{\partial_{y} \tilde{u_0}}{\tilde{u_0}} \,-\, \sum_{j=1}^{d_0} \frac{\gamma_{j}}{c_i-g_{j}^{<\delta}} \,-\, \sum_{j=1}^{d_0} \frac{\gamma_{j}\cdot A_{j,\delta}}{(c_i-\mu_{j})^2}\, \bmod I^{\delta+1}\,.$$ 

Rewriting this, using Claim \ref{lem-logDer-Id}, we get:
$$\sum_{j=1}^{d_0} \frac{\gamma_{j}}{(c_i-\mu_{j})^2} \left(v_\delta(j)+ A_{j,\delta} \right) \,\equiv\, \sum_{j=1}^{d_0} \frac{\gamma_{j}}{(c_i-\mu_j)^2}\cdot g_{j}^{=\delta} \,\bmod I^{\delta+1} \,.$$

Thus, $\sum_{j=1}^{d_0} \gamma_{j}\cdot (v_\delta(j)+ A_{j,\delta} \,-\, g_{j}^{=\delta})/(c_i-\mu_{j})^2  \,\equiv\, 0  \,\bmod I^{\delta+1} $. As we vary $i\in[d_0]$ we deduce, by Lemma \ref{lem-inv-det}, that $ v_\delta(j)+A_{j,\delta} -  g_{j}^{=\delta} \,\equiv\, 0 \, \bmod I^{\delta+1} $. Hence, $g_{j,\delta}' = g_{j,\delta-1}'+ v_\delta(j)$ $\equiv\, (g_{j}^{<\delta}+ A_{j,\delta}) + (g_{j}^{=\delta}-A_{j,\delta})$ $=\, g_{j}^{\leq\delta} \, \bmod I^{\delta+1} $. This proves it for all $j\in[d_0]$.
}

\textbf{Size analysis:} Here we give the overall process of finding factors using allRootsNI technique and analyze the circuit size needed at each step to establish the size bound of the factors. As discussed before, we need to analyze only the power series root approximation $g_i^{\leq\delta}$ or $g_{i,\delta}'$.

At the $(\delta-1)$-th step of allRootsNI process, we have a multi-output circuit (with division gates) computing $g_{i,\delta-1}'$ as a rational function, for all $i\in[d_0]$. Specifically, let us assume that $g_{i,\delta-1}'=: C_{i,\delta-1}/D_{i,\delta-1}$, where $D_{i,\delta-1}$ is invertible in $\F[[\overline{x}]]$. So, the circuit computing $g_{i,\delta-1}'$ has a division gate at the top that outputs  $C_{i,\delta-1}/D_{i,\delta-1}$. We would eliminate this division gate only in the end (see the standard Lemma \ref{lem-div-elim}).  Now we show how to construct the circuit for $g_{i,\delta}'$, given the circuits for $g_{i,\delta-1}'$.

From $v_\delta \,=\, M^{-1}W_\delta$, it is clear that there exist field elements $\beta_{ij}$ such that $v_\delta(i)= \sum_{j=1}^{d_0} \beta_{ij}W_\delta(j) = \sum_{j=1}^{d_0} \beta_{ij} \left( (\partial_{y} \tilde{f}/\tilde{f} \,-\, \partial_{y} \tilde{u_0}/\tilde{u_0})\rvert_{y=c_j} - \tilde{G}_{j,\delta}\right)$.  

Initially we precompute, for all $j\in[d_0]$, $(\partial_{y} \tilde{f}/\tilde{f} \,-\, \partial_{y} \tilde{u_0}/\tilde{u_0} )\rvert_{y=c_{j}}$: Note that $\partial_y \tilde{f}$ has $\poly(s)$ size circuit (high degree of the circuit does not matter, see Lemma \ref{lem-derivative}). Invertibility of $\tilde{f} \rvert_{y=c_{j}}$ and $\tilde{u_0}\rvert_{y=c_{j}}$ follows from the fact that we chose $c_j$'s randomly. In particular, $\tilde{f}(\overline{0},y)$, and so $\tilde{u_0}(\overline{0},y)$, have roots in $\F$ which are distinct from $c_j$, $j\in[d_0]$. Thus, $\tilde{f}(\overline{x},c_j)$ and $\tilde{u_0}(\overline{x},c_j)$ have non-zero constants and so are invertible in $\F[[\overline{x}]]$. Similarly, 
$\gamma_{\ell}/(c_{j}-g_{\ell,\delta-1}')$ exists in $\F[[\overline{x}]]$.

Thus, the matrix recurrence allows us to calculate the polynomials $C_{i,\delta}$ and $D_{i,\delta}$, given their $\delta-1$ analogues, by adding poly$(d_0)$ many wires and nodes. The precomputations costed us size $\poly(s,\delta)$. Hence, both $C_{i,\delta}$ and $D_{i,\delta}$ has poly$(s,\delta,d_0)$ sized circuit. 

We can assume we have only one division gate at the top, as for each gate $G$ we can keep track of numerator and denominator of the rational function computed at $G$, and simulate all the algebraic operations easily in this representation. When we reach precision $\delta=d_0$, we can eliminate the division gate at the top. As $D_{i,d_0}$ is a unit, we can compute its inverse using the power series inverse formula and approximate only up to degree $d_0$ (Lemma \ref{div-elm}). Finally, the circuit for the polynomial $g_{i}^{\leq d_0} \equiv C_{i,d_0}/D_{i,d_0} \bmod I^{d_0+1}$, for all $i\in[d_0]$, has size poly$(s,d_0)$.
 
Altogether, it implies that any factor of $u_1$  has a circuit of size poly$(s,d_0)$.   
\end{proof}

\subsection{Low degree factors of general circuits: Proof of Theorem \ref{thm2}}
\label{sec-pf-thm2}
\vspace{-1mm}

Here, we introduce an approach to handle the general case when rad$(f)$ has  exponential degree. We show that allowing a special kind of modular division gate  gives a small circuit for any low degree factor of $f$.

The {\em modular division} problem is to show that if $f/g$ has a representative in $\F[[\overline{x}]]$, where polynomials $f$ and $g$ can be computed by a circuit of size $s$, then $f/g \bmod \langle \overline{x}^d \rangle$ can be computed by a circuit of size poly$(sd)$. Note that if $g$ is invertible in $\F[[\overline{x}]]$, then the question of modular division can be solved using Strassen's trick of division elimination \cite{strassen1973vermeidung}. But, in our case $g$ is not invertible in $\F[[\overline{x}]]$ (though $f/g$ is well-defined).

\begin{proof}[Proof of Theorem \ref{thm2}]
As discussed before, to show size bound for an arbitrary factor (with low degree) of $f$, it is enough to show the size bound for the approximations of power series roots. 
From Theorem \ref{thm-complete-split}, $ \tilde{f}(\overline{x},y) = f(\tau \overline{x})= k\cdot \prod_{i=1}^{d_0} (y-g_i)^{\gamma_i}$, with $ g_{i}(\overline{0}) :=\mu_i$ being distinct. 

Fix an $i$ from now on.
To calculate $g_{i}^{\leq \delta}$, we iteratively use Newton iteration with multiplicity (as described in Section \ref{sec-tech}) for $\log \delta+1$ many times. We know that there are rational functions $\hat{g}_{i,t}$ such that $\hat{g}_{i,t+1} := \hat{g}_{i,t} \,-\, \gamma_{i}\cdot \frac{\tilde{f}}{\partial_y\tilde{f}}\big\rvert_{y=\hat{g}_{i,t}}$ and $\hat{g}_{i,t} \equiv g_{i} \bmod \langle \overline{x} \rangle ^{2^t}$. We compute $\hat{g}_{i,t}$'s incrementally, $0\le t\le \log \delta+1$, by a circuit with division gates. As before, $\tilde{f}$ and $\partial_y\tilde{f}$ have poly($s$) size circuits.

If $\hat{g}_{i,t}$ has $S_t$ size circuit with division, then $S_{t+1}=S_t+ O(1)$. Hence, $\hat{g}_{i,\lg \delta+1}$ has poly$(s,\log\delta)$ size circuit with division.

By keeping track of numerator and denominator of the rational function computed at each gate, we can assume that the only division gate is at the top. As the size of $\hat{g}_{i,\log \delta+1}$ was initially poly$(s,\log\delta)$ with intermediate division gates, it is easy to see that when division gates are pushed at the top, it computes $A/B$ with size of both $A$ and $B$ still poly$(s,\log\delta)$. 

Finally, a degree $\delta$ polynomial factor $h|f$ will require us to estimate $g_{i}^{\leq \delta}$ for that many $i$'s. Thus, such a factor has poly$(s\delta)$ size circuit, using a single modular division.
\end{proof}

\subsection{Closure of restricted complexity classes: Proof of Theorem \ref{thm3}}
\label{sec-pf-thm3}
\vspace{-1mm}

This subsection is dedicated towards proving closure results for certain algebraic complexity classes. In fact, for ``practical" fields like $\Q, \Q_p$, or $\F_q$ for prime-power $q$, we give efficient randomized algorithm to output the complete factorization of polynomials belonging to that class (stated as Theorem \ref{thm-betterThm3}). We use the notation $g \mid \mid f$ to denote that $g$ divides $f$ but $g^2$ does not divide $f$. Again, we denote $I:=\langle x_1,\hdots,x_n \rangle$

\begin{proof}[Proof of Theorem \ref{thm3}]
There are essentially two parts in the proof. The first part talks only about the existential closure results. In the second part, we discuss the algorithm.

{\em Proof of closure:} Given $f$ of degree $d$, we randomly shift by $\tau: x_i \mapsto x_i+y\alpha_i+\beta_i$. From Theorem \ref{thm-complete-split} we have that $\tilde{f}(\overline{x},y):= f(\tau\overline{x})$ splits like $\tilde{f}=\prod_{i=1}^{d_0}(y-g_i)^{\gamma_i}$, with $g_i(\overline{0})=:\mu_i$ being distinct. Here is the detailed size analysis of the factors of polynomials represented by various models of our interest.

\smallskip\noindent 
\textbf{Size analysis for formula:} Suppose $f$ has a formula of size $n^{O(\log n)}$. To show size bound for all the factors, it is enough to show that the approximations of the power series roots, i.e.~$g_i^{\leq d}$ has size $n^{O(\log n)}$ size formula. This follows from the reduction of factoring to approximations of power series roots.

We differentiate $\tilde{f}$ wrt $y$, $(\gamma_i-1)$ many times, so that the multiplicity of the root we want to recover becomes exactly one. The differentiation would keep the size poly$(n^{\log n})$ (Lemma \ref{lem-derivative}). Now, we have $(y-g_i) \mid \mid  \tilde{f}^{(\gamma_i-1)}$ and we can apply classical Newton iteration formula (Section \ref{sec-tech}). For all $0\le t\le \log d+1$, we compute $A_t$ and $B_t$ such that $A_t/B_t \equiv g_i \bmod I^{2^t} $. Moreover, $B_t$ is invertible in $\F[[\overline{x}]]$ ($\because$ $g_i$ is a simple root of $\tilde{f}^{(\gamma_i-1)}$). 

To implement this iteration using the formula model, each time there would be a blow up of $d^2$. Note that in a formula, there can be many copies of the same variable in the leaf nodes and if we want to feed something in that variable, we have to make equally many copies. That means we may need to make $s$ ($=\siz(f)$) many copies at each step. We claim that it can be reduced to only $d^2$ many copies.

We can pre-compute (with blow up at most poly($sd$)) all the coefficients $C_0,\ldots,C_d $ wrt $y$, given the formula of $\tilde{f}=: C_0+C_1y+\ldots+C_dy^d$ using interpolation. We can do the same for the derivative formula. For details on this interpolation trick, see \cite[Lem.5.3]{saptharishi2016survey}. Using interpolation, we can convert the formula of $\tilde f$ and its derivative to the form $C_0+C_1y+\ldots+C_dy^d$. In this modified formula, there are $O(d^2)$ many leaves labelled as $y$. So in the modified formula of the polynomial $\tilde{f}$ and in its derivative, we are computing and plugging in (for $y$) $d^2$ copies of $g_i^{<2^t}$ to get $g_i^{<2^{t+1}}$. This leads to $d^2$ blow up at each step of the iteration. 

As $B_t$'s are invertible, we can keep track of the division gates across iterations and, in the end, eliminate them causing a one-time size blow up of $\poly(sd)$ (Lemma \ref{lem-div-elim}). 

Now, assume that $\text{size}(A_{t},B_{t}) \leq S_{t}$. Then we have $ S_{t+1} \leq O(d^2S_{t}) + \text{poly}(sd)$.  Finally, we have $S_{\log d+1} = \text{poly}(sd)\cdot d^{2\log d}= \,\text{poly}(n^{\log n}) $. 

Hence, $g_i^{\leq d} \equiv A_{\log d +1}/B_{\log d+1} \bmod I^{d+1} $ has poly$(n^{\log n})$ size formula, and so does every polynomial factor of $f$ after applying $\tau^{-1}$. 

\smallskip\noindent
\textbf{Size analysis for ABP:} This analysis is similar to that of the formula model, as the size blow up in each NI iteration for differentiation, division, and truncation (to degree $\le d$) is the same as that for formulas. A noteworthy difference is that we need to eliminate division in {\em every} iteration (Lemma \ref{div-elm}) and we cannot postpone it. This leads to a blow up of $d^4$ in each step. Hence, $S_{\lg d+1} = \text{poly}(sd)\cdot d^{4\log d}= \,\text{poly}(n^{\log n}) $. 

\smallskip\noindent
\textbf{Size analysis for VNP:} Suppose $f$ can be computed by a verifier circuit of size, and witness size, $n^{O(\log n)}$. We call both the verifier circuit size and witness size as size parameter.  Now, our given polynomial $\tilde{f}$ has $n^{O(\log n)}$ size parameters. As before, it is enough to show that $g_i^{\leq d}$ has $n^{O(\log n)}$ size parameters. 

For the preprocessing (taking $\gamma_i-1$-th derivative of $\tilde{f}$ wrt $y$), the blow up in the size parameters is only poly$(n^{\log n})$. Now we analyze the blow up due to classical Newton iteration. We compute $A_t$ and $B_t$ such that $A_t/B_t \equiv g_i \bmod I^{2^t}$. Using the closure properties of VNP (discussed in Section \ref{VNP}), we see that each time there is a blow up of $d^4$. The main reason for this blow up is due to the \emph{composition} operation, as we are feeding a polynomial into another polynomial. 

Assume that the verifier circuit $\text{size}(A_{t},B_{t}) \leq S_{t}$ and witness size $\leq W_{t}$. Then we have $S_{t+1} \leq O(d^4S_{t}) + \text{poly}(n^{\log n})$. So, finally we have $S_{\log d+1} = \text{poly}(sd)\cdot d^{4\log d}= \text{poly}(n^{\log n})$.  It is clear that $g_i^{\leq d} \equiv A_{\log d +1}/B_{\log d+1} \bmod I^{d+1}$ has poly$(n^{\log n})$ size verifer circuit. Same analysis works for $W_{t}$ and witness size remains $n^{O(\log n)}$. Moreover, we get the corresponding bounds for every polynomial factor of $f$ after applying $\tau^{-1}$. 

\medskip
Before moving to the constructive part, we discuss a new method for computing gcd of two polynomials, which not only fits well in the algorithm but is also of independent interest. We recall the definition of gcd of two polynomials $f,g$ in the ring $\F[\overline{x}]$: $\text{gcd}(f,g)=:h \iff h|f$, $h|g$ and ($h'|f, h'|g \Rightarrow h'|h$). It is unique up to constant multiples.

\begin{claim}[Computing formula gcd]\label{clm-gcd}
Given two polynomials $f,g \in \F[\overline{x}]$ of degree $d$ and computed by a formula (resp.\ ABP) of size $s$. One can compute a formula (resp.~ABP) for gcd$(f,g)$, of size poly$(s,d^{\log d})$, in randomized poly$(s,d^{\log d})$ time.
\end{claim}

\claimproof{clm-gcd}{
The idea is the following. Suppose, gcd$(f,g)=: h$ is of degree $d>0$, then we will compute $h(\tau \overline{x})$ for a random map $\tau$ as in Theorem \ref{thm-complete-split}. We know wlog that $\tilde{f}:=$ $f(\tau \overline{x})= \prod_i (y-A_i)^{a_i}$ and $\tilde{g}:=$ $g(\tau \overline{x})= \prod_i (y-B_i)^{b_i}$, where $A_i,B_i \in \F[[\overline{x}]]$. Since $\F[\overline{x}]\subset \F[[\overline{x}]]$ are UFDs (Proposition \ref{prop-ufd}), we could say wlog that $h(\tau \overline{x})= \prod_{i \in S} (y-A_i)^{\min(a_i,b_i)}$, where $S=\{ i \mid A_i= B_i \}$ after possible rearrangement. Now, as $\tau$ is a random invertible map, we can assume that, for $i\ne j$, $A_i \neq B_j$ and that $A_i(\overline{0}) \neq B_j(\overline{0})$ (Lemma \ref{lem-coprime}). So, it is enough to compute $A_i^{\leq d}$ and $B_j^{\leq d}$ and compare them using evaluation at $\overline{0}$. If indeed $A_i=B_i$, then $A_i^{\leq d}=B_i^{\leq d}$. If they are not, they mismatch at the constant term itself! Hence, we know the set $S$ and so we are done once we have the power series roots with repetition.

\smallskip
Using univariate factoring, wrt $y$, we get all the multiplicities, of the roots, $a_i$ and $b_i$'s, additionally we get the corresponding starting points of classical Newton iteration, i.e.~$A_i(\overline{0})$ and $B_i(\overline{0})$'s. Using NI, one can compute $A_i^{\leq d}$ and $B_i^{\leq d}$, for all $i$. Suppose, after rearrangement of $A_i$ and $B_i$'s (if necessary), we have $A_i=B_i$ for $i \in [s]=:S$ and $A_i \neq B_j $ for $i \in [s+1,d], j\in [s+1,d]$. Lemma \ref{lem-coprime} can be used to deduce that $A_i(\overline{0}) \neq B_j(\overline{0}) $ for $i,j \in [1,d]-S$.  So, we have in $\text{gcd}(\tilde{f},\tilde{g})= \prod_{i \in S} (y-A_i)^{\min(a_i,b_i)}$: the index set $S$, the exponents and $A_i(\overline{0})$'s computed. 
 
 \emph{Size analysis:}
 We compute $A_i^{\leq d}$ and $B_i^{\leq d}$ by NI, (possibly) after making the corresponding multiplicity one by differentiation. It is clear that at each NI step there will be a multiplicative $d^2$ blow up (due to interpolation, division and truncation). There are $\log d$ iterations in NI. Altogether the truncated roots have poly$(s,d^{\log d})$ size formula (resp.\ ABP). This directly implies that gcd$(\tilde{f},\tilde{g})$ has poly$(s,d^{\log d})$ size formula (resp.\ ABP). By taking the product of the linear factors, truncating to degree $d$, and applying $\tau^{-1}$, we can compute the polynomial $\gcd(f,g)$. 
 
Randomization is needed for $\tau$ and possibly for the univariate factoring over $\F$. Also, it is important to note that $\F$ may not be algebraically closed. Then one has to go to an extension, do the algebraic operations and return back to $\F$. For details, see Section \ref{sec-not-clos}.
}

\noindent
\textbf{Randomized Algorithm. } We give the broad steps of our algorithm below. We are given $f \in \F[\overline{x}]$, of degree $d>0$, as input.
\begin{enumerate}

\item Choose $\overline{\alpha},\overline{\beta} \in_r \F^n$ and apply $\tau: x_i \rightarrow x_i + \alpha_iy+\beta_i$. Denote the transformed polynomial  $f(\tau \overline{x})$ by $\tilde{f}(\overline{x},y)$.
Wlog, from Theorem \ref{thm-complete-split}, $\tilde{f}$ has factorization of the form $\prod_{i=1}^{d_0} (y-g_i)^{\gamma_i}$, where $\mu_i:=g_i(\overline{0})$ are distinct.

\item Factorize $\tilde{f}(\overline{0},y)$ over $\F[y]$. This will give $\gamma_i$ and $\mu_i$'s.

\item Fix $i=i_0$. Differentiate $\tilde{f}$, wrt $y$, ($\gamma_{i_0}-1$) many times to make $g_{i_0}$ a simple root.
 
\item Apply Newton iteration (NI), on the differentiated polynomial, for $k:= \lceil\log (2d^2+1)\rceil$ iterations; starting with the approximation $\mu_{i_0}$ (mod $I$). We get $g_{i_0}^{< 2^k}$ at the end of the process (mod $I^{2^k}$).

\item Apply the transformation $x_i \mapsto Tx_i$ ($T$ acts as a degree-counter). Consider $\tilde{g}_{i_0} := g_{i_0}^{< 2^k}(T\overline{x})$. Solve the following homogeneous linear system of equations, over $\F[\overline{x}]$, in the unknowns $u_{ij}$ and $v_{ij}$'s, 
$$\sum_{0\le i+j <d} u_{ij} \cdot y^i T^j \,\,=\,\, (y-\tilde{g}_{i_0})\cdot \sum_{\substack{0\le i<d\\ 0\le j< 2^k} } v_{ij}\cdot y^iT^j  \,\,\bmod T^{2^k} \,.$$ 
Solve this system, using Lemma \ref{lem-linsyst}, to get a nonzero polynomial (if one exists) $u:= \sum_{0\le i+j <d} u_{ij} \cdot y^i T^j$.

\item If there is no solution, return ``$f$ is irreducible''.

\item Otherwise, find the minimal solution wrt deg$_y(u)$ by brute force (try all possible degrees wrt $y$; it is in $[d-1]$).

\item Compute $G(\overline{x},y,T):= \gcd_y (u(\overline{x},y,T), \tilde{f}(T\overline{x},y) )$ using Claim \ref{clm-gcd}. 

\item Compute $G(\overline{x},y,1)$ and transform it by $\tau^{-1}: x_i \mapsto x_i - \alpha_iy - \beta_i$, $i\in[n]$, and $y\mapsto y$. Output this as an irreducible polynomial factor of $f$.
\end{enumerate}

\begin{claim}[Existence]\label{clm-exist}	
If $f$ is reducible, then the linear system (Step 5) has a non-trivial solution.
\end{claim}

\claimproof{clm-exist}{
If $f$ is reducible, then let $f=\prod f_i^{e_i}$ be its prime factorization. Assume wlog that $y-g_{i_0} \mid \tilde{f_1}:= f_1(\tau \overline{x})$. Of course $0< \deg_y(\tilde{f_1})= \text{deg}(f_1) < d$. 

Observe that we are done by picking $u$ to be $\tilde{f_1}(T\overline{x},y)$. For, total degree of $f_1$ is $<d$, and so that of $\tilde{f_1}(T\overline{x},y)$ wrt the variables $y, T$ is $<d$. 

Moreover, $y-g_{i_0} \mid \tilde{f_1} \implies \tilde{f_1}=(y-g_{i_0})v$, for some $v\in\F[[\overline{x}]][y]$ with $\deg_y<d$. Hence, $\tilde{f_1} \equiv (y-g_{i_0}^{< 2^k})\cdot v \, \bmod I^{2^k} \implies u \equiv (y-\tilde{g}_{i_0})\cdot v(T\overline{x},y) \,\bmod T^{2^k}$. This shows the existence of a nontrivial solution of the linear system (Step 5).
}

Now, we show that if the linear system has a solution, then the solution corresponds to a non-trivial polynomial factor of $f$.

\begin{claim}[Step 8's success]\label{clm-constr}
If the linear system (Step 5) has a non-trivial solution, then $0< \deg_y G \le \deg_y u <d$. \end{claim}

\claimproof{clm-constr}{
Suppose $(u,v)$ is the solution provided by the algorithm in Lemma \ref{lem-linsyst} ($u$ being in the unknown LHS and $v$ being the unknown RHS). Consider $G=\gcd_y(u,\tilde{f}(Tx,y))$. We know that there are polynomials $a$ and $b$ such that $au+b \tilde{f}(Tx,y)=\text{Res}_y(u,\tilde{f}(Tx,y))$ (Section \ref{sec-res}). Consider $\deg_T(\text{Res}_y(u,\tilde{f}(Tx,y))$. As degree of $T$ in $u$ and $\tilde{f}(Tx,y)$ can be at most $d$, hence degree of $T$ in Resultant can be atmost $2d^2$ (Section \ref{sec-res}). Clearly, $\deg_y G \le \deg_y u< d$. If $\deg_y G =0$ then the resultant of $u, \tilde{f}(T\overline{x},y)$ wrt $y$ will be nonzero (Proposition \ref{prop-res}). Suppose the latter happens.

Now, we have $u = (y-\tilde{g}_{i_0})v \,\bmod T^{2^k}$. Since $y-g_{i_0} \mid \tilde{f}$ we get that $y-g_{i_0}(T\overline{x}) \mid \tilde{f}(T\overline{x},y)$. Assume that $\tilde{f}(Tx,y)=: (y-g_{i_0}(T\overline{x}))\cdot w$. 

Thus, we can rewrite the previous equation as: $au+b \tilde{f}(T\overline{x},y)\equiv$
$(y-\tilde{g}_{i_0})(av+bw) \equiv \text{Res}_y(u,\tilde{f}(Tx,y)) \,\bmod T^{2^k}$. Note that the latter is nonzero mod $T^{2^k}$ because the resultant is a nonzero polynomial of deg$_T$ $<2^k$.  Putting $y=\tilde{g}_{i_0}$ the LHS vanishes, but RHS does not ($\because$ it is independent of $y$). This gives a contradiction.

Thus, $\text{Res}_y(u,\tilde{f}(Tx,y) =0$. This implies that $0<\deg_y G <d$.
}

Next we show that if one takes the minimal solution $u$ (wrt degree of $y$), then it will correspond to an irreducible factor of $f$. We will use the same notation as above.

\begin{claim}[Irred.~factor]\label{clm-irred-fac}
Suppose $y-g_{i_0} \mid \tilde{f_1}$ and $f_1$ is an irreducible factor of $f$. Then, $G=c\cdot\tilde{f_1}(Tx,y)$, for $c\in\F^*$, and $\deg_y(G)=\deg_y(u)=\deg_y(f_1)$ in Step 8.
\end{claim}

\claimproof{clm-irred-fac}{
Suppose $f$ is reducible, hence as shown above, $G$ is a non-trivial factor of $\tilde{f}(T\overline{x}, y)$. Recall that $\tilde{f}(T\overline{x},y)= \prod_i (y-g_i(T\overline{x}))^{\gamma_i}$ is a factorization over $\F[[\overline{x}, T]]$. We have that $y-\tilde{g}_{i_0} \mid G$ mod $T^{2^k}$. Thus, $y-g_{i_0}(T\overline{x}) \mid G$ absolutely ($\because$ the power series ring is a UFD and use Theorem \ref{thm-complete-split}). So, $y-g_{i_0}(T\overline{x}) \mid \gcd_y(G, \tilde{f_1}(T\overline{x},y))$ over the power series ring. Since, $\tilde{f_1}(T\overline{x},y)$ is an irreducible polynomial, we can deduce that $\tilde{f_1}(T\overline{x},y) \mid G$ in the polynomial ring. So, $\deg_y(f_1)\le \deg_y(G)$.

We have $\deg_y(\tilde{f_1}(T\overline{x},y)) = \deg(f_1) =: d_1$. By the above discussion, the linear system in Step 7 will not have a solution of $\deg_y(u)$ below $d_1$. Let us consider the linear system in Step 7 that wants to find $u$ of $\deg_y=d_1$. This system has a solution, namely the one with $u:= \tilde{f_1}(T\overline{x},y) \bmod T^{2^k}$. Then, by the above claim, we will get the $G$ as well in the subsequent Step 8. This gives $\deg_y(G)\le \deg_y(u)=d_1$. With the previous inequality we get $\deg_y(G)=\deg_y(u)=\deg_y(f_1)$. In particular, $G$ and $\tilde{f_1}(Tx,y)$ are the same up to a nonzero constant multiple.
}

{\em Alternative to Claim \ref{clm-gcd}:}
The above proof (Claim \ref{clm-irred-fac}) suggests that the gcd question of Step 8 is rather special: One can just write $u$ as $\sum_{0\le i\le d_1} c_i(\overline{x},T) y^i$ and then compute the polynomial $G=\sum_{0\le i\le d_1} (c_i/c_{d_1})\cdot y^i$ as a formula (resp.~ABP), by eliminating division (Lemma \ref{div-elm}). 

Once we have the polynomial $G$ we can fix $T=1$ and apply $\tau^{-1}$ to get back the irreducible polynomial factor $f_1$ (with power series root $g_{i_0}$).

The running time analysis of the algorithm is by now routine. If we start with an $f$  computed by a formula (resp.~ABP) of size $n^{O(\log n)}$, then as observed before, one can compute $\tilde{g}_{i_0}$ which has $n^{O(\log n)}$ size formula (resp.~ABP). This takes care of Steps 1-4.

Now, solve the linear system in Steps 5-7 of the algorithm. Each entry of the matrix is a formula (resp.~ABP) size $n^{O(\log n)}$. The time complexity is similar by invoking Lemma \ref{lem-linsyst}.

Steps 8 is to compute gcd of two $n^{O(\log n)}$ size formulas (resp.~ABPs) which again can be done in $n^{O(\log n)}$ time giving a  size $n^{O(\log n)}$ formula (resp.~ABP) as discussed above. 

This completes the randomized $\poly(n^{\log n})$-time algorithm that outputs $n^{O(\log n)}$ sized factors.  

\end{proof}

\noindent {\bf Remarks. } \begin{enumerate}
\item The above results hold true for the classes $VBP(s),VF(s),VNP(s)$ for any size function $s = n^{\Omega(\log n)}$. 

\item By using a {\em reversal} technique \cite[Sec.1.1.2]{oliveira2016factors} and a modified $\tau$, our size bound can be shown to be $\poly(s,d^{\log r})$, where $r$ (resp.~$d$) is the individual-degree (resp.~degree) bound of $f$. So, when $r$ is constant, we get a factor as a $\text{poly}(s)$-size formula (resp.\ ABP). Oliveira \cite{oliveira2016factors} proved the same result for formulas. But, \cite{oliveira2016factors} used {\em slow} Newton iteration and in each iteration the method was different, owing to which the size was $\poly(s,d^r)$. 

\item By the above remark, our result can be extended to prove closure result for polynomials in VNP with {\em constant} individual degree. There are very interesting polynomials in this class, namely Permanent.
\end{enumerate}

\section{Extensions}\label{sec-extn}
\vspace{-1mm}

\subsection{Closure of approximative complexity classes}\label{sec-approx}
\vspace{-1mm}

In this section, we show that all our closure results, under factoring, can be naturally generalized to corresponding approximative algebraic complexity classes.

In computer science, the notion of approximative algebraic complexity emerged in early works on matrix multiplication (the notion of border rank, see \cite{burgisser2013algebraic}). It is also an important concept in the geometric complexity theory program (see \cite{grochow2016boundaries}). The notion of approximative complexity can be motivated through two ways, \emph{topological} and \emph{algebraic} and both the perspectives are known to be equivalent. Both allow us to talk about the \emph{convergence} $\epsilon\rightarrow0$. 

In what follows, we can see $\epsilon$ as a formal variable and $\F(\epsilon)$ as the function field. For an algebraic complexity class $C$, the approximation is defined as follows \cite[Defn.2.1]{bringmann2017algebraic}.

\begin{definition}[Approximative closure of a class \cite{bringmann2017algebraic}]
Let $C$ be an algebraic complexity class over field $\F$. A family $(f_n)$ of polynomials from $\F[\overline{x}]$ is in the {\em class $\overline{C}(\F)$}
if there are polynomials $f_{n;i}$ and a function $t:\mathbb{N} \mapsto \mathbb{N}$ such that $g_n$  is in the class $C$ over the field $\F(\epsilon)$ with $g_n(\overline{x})= f_n(\overline{x})+ \epsilon f_{n;1}(\overline{x})+ {\epsilon}^2 f_{n;2}(x)+ \ldots+ {\epsilon}^{t(n)} f_{n;t(n)}(\overline{x})$.
\end{definition}

The above definition can be used to define closures of classes like VF, VBP, VP, VNP which are denoted as $\overline{\text{VF}}$, $\overline{\text{VBP}}$, $\overline{\text{VP}}$, $\overline{\text{VNP}}$ respectively. In these cases one can assume wlog that the degrees of $g_n$ and $f_{n;i}$ are $\poly(n)$.

{\em Following B\"urgisser \cite{burgisser2001complexity}:-}
Let $K:=\F(\epsilon)$ be the rational function field in variable $\epsilon$ over the field $\F$. Let $R$ denote the subring of $K$ that consists of rational functions defined in $\epsilon=0$. Eg.~$1/\epsilon\notin R$ but $1/(1+\epsilon)\in R$.

\begin{definition}\cite[Defn.3.1]{burgisser2001complexity}
Let $f \in \F[x_1,\ldots,x_n]$. The {\em approximative complexity} $\overline{\siz}(f)$ is the smallest number $r$, such that there exists $F$ in $R[x_1,\ldots,x_n]$ satisfying $F\rvert_{\epsilon=0}=f$ and circuit size of $F$ over {\em constants} $K$ is $\leq r$.
\end{definition}

Note that the circuit of $F$ may be using division by $\epsilon$ implicitly in an intermediate step. So, we cannot simply assign $\epsilon=0$ and get a circuit free of $\epsilon$. Also, the degree involved can be arbitrarily large wrt $\epsilon$. Thus, potentially $\overline{\siz}(f)$ can be smaller than $\siz(f)$. 

Using this new notion of size one can define the analogous class $\overline{\text{VP}}$. It is known to be closed under factors \cite[Thm.4.1]{burgisser2001complexity}.
The idea is to  work over $\F(\epsilon)$, instead of working over $\F$, and use Newton iteration to approximate power series roots. Note that in the case of $\overline{\text{VF}}$, $\overline{\text{VBP}}$, $\overline{\text{VP}}$ and $\overline{\text{VNP}}$ the polynomials have $\poly(n)$ degree. So, by using repeated differentiation, we can assume the power series root (of $\tilde{f}:= f(\tau\overline{x})$) to be simple (i.e.~multiplicity$=1$) and apply classical NI. We need to carefully analyze the implementation of this idea. 

\medskip\noindent {\bf Root finding using NI over $K$. }
For degree-$d$ $f\in\F[\overline{x}]$ if $\overline{\siz}(f)=s$ then: $\exists F \in R[\overline{x}]$ with a size $s$ circuit satisfying $F\rvert_{\epsilon=0}=f$. 
The degree of $F$ wrt $\overline{x}$ may be greater than $d$. In that case we can extract the part up to degree $d$ and truncate the rest \cite[Prop.3.1]{burgisser2004complexity}. So wlog $\deg_{\overline{x}}(F)= \deg(f)$.

By applying a random $\tau$ (using constants $\F$) we can assume that $\tilde{F} := F(\tau\overline{x}) \in R[\overline{x},y]$ is {\em monic} (i.e.~leading-coefficient, wrt $y$ in $\tilde{F}$, is invertible in $R$). Otherwise, $\deg_y(\tilde{F}) = \deg_y(\tilde{f}) = \deg_{\overline{x}}(f)$ will decrease on substituting $\epsilon=0$ contradicting $F\rvert_{\epsilon=0}=f$. Wlog, we can assume that the leading-coefficient of $\tilde{F}$ wrt $y$ is $1$ and the $y$-monomial's degree is $d$. From now on we have $\tilde{F}\rvert_{\epsilon=0}= \tilde{f}$ and both have their leading-coefficients $1$ wrt $y$.

Let $\mu$ be a root of $\tilde{f}(\overline{0},y)$ of multiplicity one (as discussed before). Since $\tilde{F}(\overline{0},y)\equiv \tilde{f}(\overline{0},y) \bmod \epsilon$, we can build a power series root $\mu(\epsilon)\in\F[[\epsilon]]$ of $\tilde{F}(\overline{0},y)$ using NI, with $\mu$ as the starting point. But $\mu(\epsilon)$ may not converge in $K$. To overcome this obstruction \cite{burgisser2001complexity} devised a clever trick. 

Define $\hat{F}:= \,\tilde{F}(\overline{x}, y+\mu+\epsilon) \,-\, \tilde{F}(\overline{0}, \mu+\epsilon)$. Note that $(\overline{0}, 0)$ is a simple root of $\hat{F}(\overline{x}, y)$ \cite[Eqn.5]{burgisser2004complexity}. So, a power series root $y_\infty$ of $\hat{F}$ can be built iteratively by classic NI (Lemma \ref{lem-NI}):
$$  y_{t+1} \;\,:=\,\; y_t \,-\, \frac{\hat{F}}{\partial_y \hat{F}}\bigg\rvert_{y=y_t} \,. $$
Where, $y_\infty \equiv y_{t}\, \bmod \langle {\bar {x}}\rangle^{2^t}$. 
One can easily prove that $y_t$ is defined over the coefficient field $K$, using induction on $t$. 

Note that $\hat{F}\rvert_{\epsilon=0}= \tilde{f}(\overline{x}, y+\mu) \,-\, \tilde{f}(\overline{0}, \mu) \,=\, \tilde{f}(\overline{x}, y+\mu)$. So, $y_\infty$ is associated with a root of $\tilde{f}$ as well. This implies that by using several such roots $y_\infty$, we can get an appropriate product $\hat{G}\in R[\overline{x},y]$, such that an actual polynomial factor of $\tilde{f}$ (over field $\F$) equals $\hat{G}\rvert_{\epsilon=0}$. 


The above process, when combined with the first part of the proof of Theorem \ref{thm3}, does imply:
 
\begin{theorem}[Approximative factors]\label{thm-vf-bar}
The approximative complexity classes $\overline{\text{VF}}(n^{\log n})$,\\ $\overline{\text{VBP}}(n^{\log n}) $ and $\overline{\text{VNP}}(n^{\log n}) $ are closed under factors.
\end{theorem}

The same question for the classes $\overline{\text{VF}}$, $\overline{\text{VBP}}$ and $\overline{\text{VNP}}$ we leave as an open question. (Though, for the respective bounded individual-degree polynomials we have the result as before.)

\subsection{When field $\F$ is not algebraically closed}
\label{sec-not-clos}
\vspace{-1mm}

We show that all our results ``partially'' hold true for fields $\F$ which are not algebraically closed. The common technique used in all the proofs is the structural result (Theorem \ref{thm-complete-split}) which talks about power series roots with respect to $y$. Recall that we use a random linear map $\tau: x_i \mapsto x_i+ \alpha_i y + \beta_i$, where $\alpha_i,\beta_i \in_r \F$, to make the input polynomial $f$ monic in $y$ and the individual degree of $y$ equal to $d:=\deg(f)$.
If we set all the variables to zero except $y$, we get a univariate polynomial $\tilde{f}(\overline{0},y)$ whose roots we are interested in finding explicitly. 

The other common technique in our proofs is the classical NI, which starts with just one field root, say $\mu_1$ of $\tilde{f}(\overline{0},y)$, and builds the full power series on it.
Let $E\subsetneq \overline{\F}$ be the smallest field where a root $\mu_1$ can be found. Say, $g|\tilde{f}_1(\overline{0},y)$ is the minimal polynomial for $\mu_1$. The degree of the extension $E:=\F[z]/(g(z))$ is at most $d$. So, computations over $E$ can be done efficiently.
The key idea is to view $E/\F$ as a vector space and simulate the arithmetic operations over $E$ by operations over $\F$. The details of this kind of simulation can be seen in \cite{von2013modern}. In circuits it means that we make $\deg(E/\F)$ copies of each gate and simulate the algebraic operations on these `tuples' following the $\F$-module structure of $E[\overline{x}]$.

Once we have found all the power series roots of $\tilde{f}(\overline{x},y)$ over $E[[\overline{x}]]$, say starting from each of the conjugates $\mu_1,\ldots, \mu_i\in E$, it is easy to get a polynomial factor in $E[\overline{x},y]$. This factor will not be in $\F[\overline{x},y]$, unless $E$ is a splitting field of $\tilde{f}_1(\overline{0},y)$. A more practical method is: While solving the linear system over $E$ in Steps 5-7 (Algorithm in Theorem \ref{thm3}) we can demand an $\F$-solution $u$. Basically, at the level of algorithm in Lemma \ref{lem-linsyst}, we can rewrite the linear system $Mw = (\sum_{0\le i\le d} M_i z^i )\cdot w = 0$ as $M_iw=0$ ($i\in[0,d]$), where the entries of the matrix $M_i$ are given as formulas (resp.~ABP) computing a $\poly(n)$ degree polynomial in $\F[\overline{x}]$. This way we get the desired $\F$-solution $u$. 
Then, Steps 8-9 will yield an irreducible polynomial factor of $f$ in $\F[\overline{x},y]$. This sketches the following more practical version of Theorem \ref{thm3}.

\begin{theorem}\label{thm-betterThm3}
For $\F$ a number field, a local field, or a finite field (with characteristic $>\deg(f)$), there exists a randomized $\poly(sn^{\log n})$-time algorithm that: for a given $n^{O(\log n)}$ size formula (resp.\ ABP) $f$ of $\poly(n)$-degree and bitsize $s$, outputs $n^{O(\log n)}$ sized formulas (resp.\ ABPs) corresponding to each of the nontrivial factors of $f$. 
\end{theorem}

Note that over these fields there are famous randomized algorithms to factor univariate polynomials in the base case, see \cite[Part III]{von2013modern} \& \cite{pauli2001factoring}.

\smallskip
The allRootsNI method in Theorem \ref{thm1} seems to require all the roots $\mu_i, i\in[d_0]$, to begin with. Let $\tilde{u}_1:= \rad(u_1(\tau\overline{x}))$. Since $\mu_i$'s are in the {\em splitting field} $E\subset\overline{\F}$ of $\rad(\tilde{u}_1(\overline{0},y))$, we do indeed get the size bound of the power series roots $g_i^{\le d_0}$ of $\tilde{u}_1$ assuming the constants from $E$. As seen in the proof, any irreducible polynomial factor $\tilde{h}_i:= h_i(\tau\overline{x})$ of $\rad(\tilde{u}_1)$ is some product of these $(y-g_i^{\le d_0})$'s mod $I^{d_0+1}$.  So, for the polynomial $\tilde{h}_i$ in $\F[\overline{x},y]$ we get a size upper bound over constants $E$. We leave it as an open question to transfer it over constants $\F$ (note: $E/\F$ can be of exponential degree). 


\subsection{Multiplicity issue in prime characteristic}\label{sec-prime}
\vspace{-1mm}

The main obstruction in prime characteristic is when the multiplicity of a factor is a $p$-multiple, where $p\ge2$ is the characteristic of $\F$. In this case, all versions of Newton iteration fail. This is because the derivative of a $p$-powered polynomial vanishes. 
When $p$ is greater than the degree of the input polynomial, these problems do not occur, so all our theorems hold (also see Section \ref{sec-not-clos}). 

When $p$ is smaller than the degree of the input polynomial in Theorem \ref{thm3}, adapting an idea from \cite[Sec.3.1]{kopparty2015equivalence}, we claim that we can give $n^{O(\lambda\log n)}$-sized formula (resp.~ABP) for the $p^{e_i}$-th power of $f_i$, where $f_i$ is a factor of $f$ whose multiplicity is divisible exactly by $p^{e_i}$, and $\lambda$ is the number of distinct $p$-powers that appear. 

Note that presently it is an open question to show that: If a circuit (resp.\ formula resp.\ ABP) of size $s$ computes $f^p$, then $f$ has a poly($sp$)-sized circuit (resp.\ formula resp.\ ABP).
   
Theorem \ref{thm3} can be extended to all characteristic as follows.
\begin{theorem}
Let $\F$ be of characteristic $p\ge2$. 
Suppose the $\poly(n)$-degree polynomial given by a $n^{O(\log n)}$ size formula (resp.~ABP) factors into irreducibles as $f(\overline{x})= \prod_{i} f_i^{{p^{e_i}}j_i}$, where $p \nmid j_i$. Let $\lambda:= \#\{e_i|i\}$.

Then, there is a poly$(n^{\lambda\log n})$-size formula (resp.\ ABP) computing $f_i^{p^{e_i}}$ over $\overline{\F}_p$. 
\end{theorem}

\begin{proof}[Proof sketch]
Note that $\lambda = O(\log_p n)$.

Let the transformed polynomial of degree $d$ split into power series roots as follows:
$\tilde{f}:=f(\tau\overline{x},y) = \prod_{i=1}^{d_0}(y-g_i)^{\gamma_i}$. 

{\em $p\nmid \gamma_i$: }
If $g_i$ is such that $p\nmid\gamma_i$, then we can find the corresponding power series roots using Newton iteration and recover all such factors. After recovering all such irreducible polynomial factors, we can divide $\tilde{f}$ by their product. Let $G:= \tilde{f} \big/ \prod_{p \nmid \gamma_i} (y-g_i)^{\gamma_i}$. Clearly, $G$ is now a $p$-power polynomial.  

{\em $p\mid \gamma_i$: }
Computing the highest power of $p$ that divides the exponent of $G$ (given by a formula resp.~ABP) is easy. First, write the polynomial as $G=c_{0}+c_{1}y+....+c_{d}y^d$ using  interpolation.  Note that it is a $p^e$-th power iff: $c_{i}=0$ whenever $p^e\nmid i$, and $p^{e+1}$ does not have this property. After computing the right value of $p^e$, we can reduce factoring to the case of a non-$p$-power. 

Rewrite $G$ as $\hat{G}:= \sum_{p^e|i} c_i(\overline{x})\cdot y^{i/p^e}$, i.e.~replacing $y^{p^e}$ by $y$. Clearly, $g$ is an irreducible factor of  $G$ iff $\hat{g}$ is an irreducible factor of $\hat{G}$.

We can now apply NI to find the roots of $\tilde{G}$, that have multiplicity coprime to $p$. Divide by their product and then repeat the above.

\emph{Size analysis. }
If $G$ can be computed by a size $s$ formula (resp.\ ABP), $\hat{G}$ can be computed by a size $O(d^2s)$ formula (resp.\ ABP). Similarly, a single division gate leads to a blow up by a factor of $O(d^2)$.
The number of times we need to eliminate division is at most $\lambda\log d$. 
So the overall size is $n^{O(\lambda\log n)}$. 

However, the splitting field $E$ where we get all the roots of $\tilde{f}(\overline{0},y)$ may be of degree $\Omega(d!)$. So, we leave the efficiency aspects of the algorithm as an open question.
\end{proof}
 
\medskip \noindent
{\bf High degree case. }
Note that the above idea cannot be implemented efficiently in the case of high degree circuits. Still we can extend our Theorem \ref{thm1} using allRootsNI. The key observation is that the allRootsNI formula still holds but the summands that appear are exactly the ones corresponding to $g_i$ with $\gamma_i\ne0 \bmod p$.

This motivates the definition of a partial radical:  $\text{rad}_p(f):= \prod_{p \nmid e_i} f_i$, if the prime factorization of $f$ is $\prod_i f_i^{e_i}$.

\begin{theorem}
Let $\F$ be of characteristic $p\ge2$. 
Let $f=u_0u_1$ such that size($f$)$ + $size($u_0$) $\leq s$. Any factor of  $\text{rad}_p(u_1)$ has size poly($s+ \deg(\rad_p(u_1) )$) over $\overline{\F}$.
\end{theorem}

{\em Proof idea: } Observe that the roots with multiplicity divisible by $p$ do not contribute to the allRootsNI process. So, the process works with $\text{rad}_p(u_1)$ and the linear algebra complexity involved is polynomial in its degree.

\section{Conclusion}
\vspace{-1mm}

The old {\em Factors conjecture} states that for a nonzero polynomial $f$: $g \mid f \implies \siz(g)\le \poly(\siz(f),\deg(g))$. Motivated by Theorem \ref{thm1}, we would like to strengthen it to:

\begin{conjecture}[radical]
For a nonzero $f$: $\text{min}\{\text{deg}(\rad(f)), \siz(\rad(f)) \}\le \poly(\siz(f))$.
\end{conjecture}

Is the Radical conjecture true if we replace size by $\overline{\siz}$?

\smallskip\noindent
In low degree regime also there are many open questions. Can we identify a class ``below'' VP that is closed under factoring? We conclude with some interesting questions. 

\begin{enumerate}
\item Are $\text{VF},\text{VBP}$ or $\text{VNP}$ closed under factoring? We might consider Theorem \ref{thm3} as a positive evidence. Additionally, note that these classes are already closed under $e$-th root taking. This is easy to see using the classic Taylor series of $(1+f)^{1/e}$, where $f\in\langle\overline{x}\rangle$.

In fact, what about the classes which are contained in $VF(n^{\log n})$ but larger than $VF$. For example, is VF$(n^{\log \log n})$ closed under factoring?

\item Can we find a suitable analog of Strassen's (non-unit) division elimination for high degree circuits? This, by Theorem \ref{thm2}, will resolve Factors conjecture.

\item Our results weaken when $\F$ is not algebraically closed or has a small prime characteristic (Sections \ref{sec-not-clos}, \ref{sec-prime}). Can we strengthen the methods to work for all $\F$?
\end{enumerate}

\noindent
{\bf Acknowledgements. }
We thank Rafael Oliveira for extensive discussions regarding his works and about circuit factoring in general. In particular, we used his suggestions about VNP and $\overline{\text{VP}}$ in our results. We are grateful to the organizers of WACT'16 (Tel Aviv, Israel) and Dagstuhl'16 (Germany) for the stimulating workshops. P.D.~would like to thank CSE, IIT Kanpur for the hospitality. N.S.~thanks the funding support from DST (DST/SJF/MSA-01/2013-14). We thank Manindra Agrawal, Sumanta Ghosh, Partha Mukhopadhyay, Thomas Thierauf and Nikhil Balaji for the discussions.

\medskip

\bibliographystyle{alpha}
\bibliography{bibliography}

\newcommand{\etalchar}[1]{$^{#1}$}
\begin{thebibliography}{DMM{\etalchar{+}}14}

\bibitem[AFGS17]{AFGS17}
Manindra Agrawal, Michael Forbes, Sumanta Ghosh, and Nitin Saxena.
\newblock Small hitting-sets for tiny arithmetic circuits or: How to turn bad
  designs into good.
\newblock Technical report,
  https://www.cse.iitk.ac.in/users/nitin/research.html, 2017.

\bibitem[AV08]{agrawal2008arithmetic}
Manindra Agrawal and V~Vinay.
\newblock Arithmetic circuits: A chasm at depth four.
\newblock In {\em Foundations of Computer Science, 2008. FOCS'08. IEEE 49th
  Annual IEEE Symposium on}, pages 67--75. IEEE, 2008.

\bibitem[AW11]{allender2011power}
Eric Allender and Fengming Wang.
\newblock On the power of algebraic branching programs of width two.
\newblock {\em Automata, Languages and Programming}, pages 736--747, 2011.

\bibitem[BCS13]{burgisser2013algebraic}
Peter B{\"u}rgisser, Michael Clausen, and Amin Shokrollahi.
\newblock {\em Algebraic complexity theory}, volume 315.
\newblock Springer Science \& Business Media, 2013.

\bibitem[BIZ17]{bringmann2017algebraic}
Karl Bringmann, Christian Ikenmeyer, and Jeroen Zuiddam.
\newblock On algebraic branching programs of small width.
\newblock In {\em 32nd Computational Complexity Conference, {CCC} 2017, July
  6-9, 2017, Riga, Latvia}, pages 20:1--20:31, 2017.

\bibitem[BOC92]{ben1992computing}
Michael Ben-Or and Richard Cleve.
\newblock Computing algebraic formulas using a constant number of registers.
\newblock {\em SIAM Journal on Computing}, 21(1):54--58, 1992.

\bibitem[BSS89]{blum1989theory}
Lenore Blum, Mike Shub, and Steve Smale.
\newblock On a theory of computation and complexity over the real numbers:
  {NP}-completeness, recursive functions and universal machines.
\newblock {\em Bulletin (New Series) of the American Mathematical Society},
  21(1):1--46, 1989.

\bibitem[B{\"u}r01]{burgisser2001complexity}
Peter B{\"u}rgisser.
\newblock The complexity of factors of multivariate polynomials.
\newblock In {\em In Proc. 42th IEEE Symp. on Foundations of Comp. Science},
  2001.

\bibitem[B{\"u}r04]{burgisser2004complexity}
Peter B{\"u}rgisser.
\newblock The complexity of factors of multivariate polynomials.
\newblock {\em Foundations of Computational Mathematics}, 4(4):369--396, 2004.
\newblock (Preliminary version in FOCS 2001).

\bibitem[B{\"u}r13]{burgisser2013completeness}
Peter B{\"u}rgisser.
\newblock {\em Completeness and reduction in algebraic complexity theory},
  volume~7.
\newblock Springer Science \& Business Media, 2013.

\bibitem[CRS96]{courant1996mathematics}
Richard Courant, Herbert Robbins, and Ian Stewart.
\newblock {\em What is Mathematics?: an elementary approach to ideas and
  methods}.
\newblock Oxford University Press, USA, 1996.

\bibitem[DB08]{dahlquist2008numerical}
Germund Dahlquist and {\AA}ke Bj{\"o}rck.
\newblock Numerical methods in scientific computing, volume {I}.
\newblock {\em Society for Industrial and Applied Mathematics}, 2008.

\bibitem[DMM{\etalchar{+}}14]{DurandMMRS14}
Arnaud Durand, Meena Mahajan, Guillaume Malod, Nicolas de~Rugy{-}Altherre, and
  Nitin Saurabh.
\newblock Homomorphism polynomials complete for {VP}.
\newblock In {\em 34th International Conference on Foundation of Software
  Technology and Theoretical Computer Science, {FSTTCS}}, pages 493--504, 2014.

\bibitem[DSY09]{dvir2009hardness}
Zeev Dvir, Amir Shpilka, and Amir Yehudayoff.
\newblock Hardness-randomness tradeoffs for bounded depth arithmetic circuits.
\newblock {\em SIAM Journal on Computing}, 39(4):1279--1293, 2009.
\newblock (Preliminary version in STOC'08).

\bibitem[FS15]{forbes2015complexity}
Michael~A Forbes and Amir Shpilka.
\newblock Complexity theory column 88: Challenges in polynomial factorization.
\newblock {\em ACM SIGACT News}, 46(4):32--49, 2015.

\bibitem[FSTW16]{forbes2016proof}
Michael~A Forbes, Amir Shpilka, Iddo Tzameret, and Avi Wigderson.
\newblock Proof complexity lower bounds from algebraic circuit complexity.
\newblock In {\em Proceedings of the 31st Conference on Computational
  Complexity}, page~32. Schloss Dagstuhl--Leibniz-Zentrum fuer Informatik,
  2016.

\bibitem[GMQ16]{grochow2016boundaries}
Joshua~A. Grochow, Ketan~D. Mulmuley, and Youming Qiao.
\newblock Boundaries of {VP} and {VNP}.
\newblock In {\em 43rd International Colloquium on Automata, Languages, and
  Programming (ICALP 2016)}, volume~55, pages 34:1--34:14, 2016.

\bibitem[GMS{\etalchar{+}}86]{gill1986projected}
Philip~E Gill, Walter Murray, Michael~A Saunders, John~A Tomlin, and Margaret~H
  Wright.
\newblock {On projected Newton barrier methods for linear programming and an
  equivalence to Karmarkar’s projective method}.
\newblock {\em Mathematical programming}, 36(2):183--209, 1986.

\bibitem[Gro15]{G15}
Joshua~A Grochow.
\newblock Unifying known lower bounds via geometric complexity theory.
\newblock {\em computational complexity}, 24(2):393--475, 2015.

\bibitem[GS98]{guruswami1998improved}
Venkatesan Guruswami and Madhu Sudan.
\newblock Improved decoding of reed-solomon and algebraic-geometric codes.
\newblock In {\em Foundations of Computer Science, 1998. Proceedings. 39th
  Annual Symposium on}, pages 28--37. IEEE, 1998.

\bibitem[GTZ88]{gianni1988grobner}
Patrizia Gianni, Barry Trager, and Gail Zacharias.
\newblock Gr{\"o}bner bases and primary decomposition of polynomial ideals.
\newblock {\em Journal of Symbolic Computation}, 6(2):149--167, 1988.

\bibitem[IKRS12]{IKRS12}
G{\'a}bor Ivanyos, Marek Karpinski, Lajos R{\'o}nyai, and Nitin Saxena.
\newblock Trading grh for algebra: algorithms for factoring polynomials and
  related structures.
\newblock {\em Mathematics of Computation}, 81(277):493--531, 2012.

\bibitem[Jan11]{jansen2011extracting}
Maurice~J Jansen.
\newblock Extracting roots of arithmetic circuits by adapting numerical
  methods.
\newblock In {\em 2nd Symposium on Innovations in Computer Science (ICS 2011)},
  pages 87--100, 2011.

\bibitem[Kal85]{Kaltofen85}
Erich Kaltofen.
\newblock Computing with polynomials given by straight-line programs {I:}
  greatest common divisors.
\newblock In {\em Proceedings of the 17th Annual {ACM} Symposium on Theory of
  Computing, May 6-8, 1985, Providence, Rhode Island, {USA}}, pages 131--142,
  1985.

\bibitem[Kal86]{Kaltofen86}
Erich Kaltofen.
\newblock Uniform closure properties of p-computable functions.
\newblock In {\em Proceedings of the 18th Annual {ACM} Symposium on Theory of
  Computing, May 28-30, 1986, Berkeley, California, {USA}}, pages 330--337,
  1986.

\bibitem[Kal87]{kaltofen1987single}
Erich Kaltofen.
\newblock Single-factor hensel lifting and its application to the straight-line
  complexity of certain polynomials.
\newblock In {\em Proceedings of the nineteenth annual ACM symposium on Theory
  of computing}, pages 443--452. ACM, 1987.

\bibitem[Kal89]{kaltofen1989factorization}
Erich Kaltofen.
\newblock Factorization of polynomials given by straight-line programs.
\newblock {\em Randomness and Computation}, 5:375--412, 1989.

\bibitem[Kal90]{kaltofen1990polynomial}
Erich Kaltofen.
\newblock Polynomial factorization 1982-1986.
\newblock {\em Dept. of Comp. Sci. Report}, pages 86--19, 1990.

\bibitem[Kal92]{kaltofen1992polynomial}
Erich Kaltofen.
\newblock Polynomial factorization 1987--1991.
\newblock {\em LATIN'92}, pages 294--313, 1992.

\bibitem[Kay11]{kayal2011efficient}
Neeraj Kayal.
\newblock Efficient algorithms for some special cases of the polynomial
  equivalence problem.
\newblock In {\em Proceedings of the twenty-second annual ACM-SIAM symposium on
  Discrete Algorithms}, pages 1409--1421. Society for Industrial and Applied
  Mathematics, 2011.

\bibitem[Kem10]{kemper2010course}
Gregor Kemper.
\newblock {\em A course in Commutative Algebra}, volume 256.
\newblock Springer Science \& Business Media, 2010.

\bibitem[KI03]{kabanets2003derandomizing}
Valentine Kabanets and Russell Impagliazzo.
\newblock Derandomizing polynomial identity tests means proving circuit lower
  bounds.
\newblock In {\em Proceedings of the thirty-fifth annual ACM symposium on
  Theory of computing}, pages 355--364. ACM, 2003.

\bibitem[KK08]{kaltofen2008expressing}
Erich Kaltofen and Pascal Koiran.
\newblock Expressing a fraction of two determinants as a determinant.
\newblock In {\em Proceedings of the twenty-first international symposium on
  Symbolic and algebraic computation}, pages 141--146. ACM, 2008.

\bibitem[KP12]{book-KP12}
Steven~G Krantz and Harold~R Parks.
\newblock {\em The implicit function theorem: history, theory, and
  applications}.
\newblock Springer Science \& Business Media, 2012.

\bibitem[KS06]{KS06}
Neeraj Kayal and Nitin Saxena.
\newblock Complexity of ring morphism problems.
\newblock {\em computational complexity}, 15(4):342--390, 2006.

\bibitem[KS09]{karnin2009reconstruction}
Zohar~S Karnin and Amir Shpilka.
\newblock Reconstruction of generalized depth-3 arithmetic circuits with
  bounded top fan-in.
\newblock In {\em Computational Complexity, 2009. CCC'09. 24th Annual IEEE
  Conference on}, pages 274--285. IEEE, 2009.

\bibitem[KS16]{KS16}
Mrinal Kumar and Shubhangi Saraf.
\newblock Arithmetic circuits with locally low algebraic rank.
\newblock In {\em 31st Conference on Computational Complexity, {CCC} 2016, May
  29 to June 1, 2016, Tokyo, Japan}, pages 34:1--34:27, 2016.

\bibitem[KSS15]{kopparty2015equivalence}
Swastik Kopparty, Shubhangi Saraf, and Amir Shpilka.
\newblock Equivalence of polynomial identity testing and polynomial
  factorization.
\newblock {\em computational complexity}, 24(2):295--331, 2015.

\bibitem[Lec02]{lecerf2002quadratic}
Gr{\'e}goire Lecerf.
\newblock Quadratic newton iteration for systems with multiplicity.
\newblock {\em Foundations of Computational Mathematics}, 2(3):247--293, 2002.

\bibitem[LLMP90]{lenstra1990number}
Arjen~K Lenstra, Hendrik~W Lenstra, Mark~S Manasse, and John~M Pollard.
\newblock The number field sieve.
\newblock In {\em Proceedings of the twenty-second annual ACM symposium on
  Theory of computing}, pages 564--572. ACM, 1990.

\bibitem[LN97]{LN}
Rudolph Lidl and Harald Niederreiter.
\newblock {\em Finite Fields}.
\newblock Cambridge University Press, Cambridge, UK, 1997.

\bibitem[LS78]{lipton1978evaluation}
Richard~J Lipton and Larry~J Stockmeyer.
\newblock Evaluation of polynomials with super-preconditioning.
\newblock {\em Journal of Computer and System Sciences}, 16(2):124--139, 1978.

\bibitem[Mah14]{mahajan2014algebraic}
Meena Mahajan.
\newblock Algebraic complexity classes.
\newblock In {\em Perspectives in Computational Complexity}, pages 51--75.
  Springer, 2014.

\bibitem[Mul12a]{Mul12b}
Ketan~D. Mulmuley.
\newblock The {GCT} program toward the {P vs. NP} problem.
\newblock {\em Commun. ACM}, 55(6):98--107, June 2012.

\bibitem[Mul12b]{mul12}
Ketan~D. Mulmuley.
\newblock Geometric complexity theory {V}: {E}quivalence between blackbox
  derandomization of polynomial identity testing and derandomization of
  {N}oether's normalization lemma.
\newblock In {\em FOCS}, pages 629--638, 2012.

\bibitem[Mul17]{mulmuley2017geometric}
Ketan Mulmuley.
\newblock Geometric complexity theory {V}: Efficient algorithms for {N}oether
  normalization.
\newblock {\em Journal of the American Mathematical Society}, 30(1):225--309,
  2017.

\bibitem[MV97]{mahajan1997combinatorial}
Meena Mahajan and V~Vinay.
\newblock A combinatorial algorithm for the determinant.
\newblock In {\em SODA}, pages 730--738, 1997.

\bibitem[New69]{newton1669}
Isaac Newton.
\newblock De analysi per aequationes numero terminorum infinitas [on analysis
  by infinite series] (in latin).
\newblock 1669.
\newblock (published in 1711 by William Jones).

\bibitem[Oli16]{oliveira2016factors}
Rafael Oliveira.
\newblock Factors of low individual degree polynomials.
\newblock {\em Computational Complexity}, 2(25):507--561, 2016.
\newblock (Preliminary version in CCC'15).

\bibitem[OR00]{ortega2000iterative}
James~M Ortega and Werner~C Rheinboldt.
\newblock {\em Iterative solution of nonlinear equations in several variables}.
\newblock SIAM, 2000.

\bibitem[Pau01]{pauli2001factoring}
Sebastian Pauli.
\newblock Factoring polynomials over local fields.
\newblock {\em Journal of Symbolic Computation}, 32(5):533--547, 2001.

\bibitem[Pla77a]{plaisted1977new}
David~Alan Plaisted.
\newblock New {NP-hard and NP-complete} polynomial and integer divisibility
  problems.
\newblock In {\em Foundations of Computer Science, 18th Annual Symposium on},
  pages 241--253. IEEE, 1977.

\bibitem[Pla77b]{plaisted1977sparse}
David~Alan Plaisted.
\newblock Sparse complex polynomials and polynomial reducibility.
\newblock {\em Journal of Computer and System Sciences}, 14(2):210--221, 1977.

\bibitem[PSS16]{PSS16}
Anurag Pandey, Nitin Saxena, and Amit Sinhababu.
\newblock Algebraic independence over positive characteristic: New criterion
  and applications to locally low algebraic rank circuits.
\newblock In {\em 41st International Symposium on Mathematical Foundations of
  Computer Science, {MFCS} 2016, August 22-26, 2016 - Krak{\'{o}}w, Poland},
  pages 74:1--74:15, 2016.

\bibitem[Sap16]{saptharishi2016survey}
Ramprasad Saptharishi.
\newblock A survey of lower bounds in arithmetic circuit complexity.
\newblock {\em URL https://github. com/dasarpmar/lowerbounds-survey/releases.
  Version}, 3(0), 2016.

\bibitem[Sch77]{schnorr1977improved}
Claus-Peter Schnorr.
\newblock Improved lower bounds on the number of multiplications/divisions
  which are necessary to evaluate polynomials.
\newblock In {\em International Symposium on Mathematical Foundations of
  Computer Science}, pages 135--147. Springer, 1977.

\bibitem[Sch80]{Sch80}
J.~T. Schwartz.
\newblock Fast probabilistic algorithms for verification of polynomial
  identities.
\newblock {\em J. ACM}, 27(4):701--717, October 1980.

\bibitem[Sin16]{sinha2016reconstruction}
Gaurav Sinha.
\newblock Reconstruction of real depth-3 circuits with top fan-in 2.
\newblock In {\em 31st Conference on Computational Complexity}, 2016.

\bibitem[Str73]{strassen1973vermeidung}
Volker Strassen.
\newblock Vermeidung von divisionen.
\newblock {\em Journal f{\"u}r die reine und angewandte Mathematik},
  264:184--202, 1973.

\bibitem[Sud97]{sudan1997decoding}
Madhu Sudan.
\newblock Decoding of reed solomon codes beyond the error-correction bound.
\newblock {\em Journal of complexity}, 13(1):180--193, 1997.

\bibitem[SY10]{shpilka2010arithmetic}
Amir Shpilka and Amir Yehudayoff.
\newblock Arithmetic circuits: A survey of recent results and open questions.
\newblock {\em Foundations and Trends{\textregistered} in Theoretical Computer
  Science}, 5(3--4):207--388, 2010.

\bibitem[Tay15]{taylor1715}
Brook Taylor.
\newblock Methodus incrementorum directa et inversa [direct and reverse methods
  of incrementation] (in latin).
\newblock 1715.
\newblock (Translated into English in Struik, D. J. (1969). A Source Book in
  Mathematics 1200–1800. Cambridge, Massachusetts: Harvard University Press.
  pp. 329–332.).

\bibitem[Val79]{V79}
Leslie~G. Valiant.
\newblock Completeness classes in algebra.
\newblock In {\em Proceedings of the 11h Annual {ACM} Symposium on Theory of
  Computing, April 30 - May 2, 1979, Atlanta, Georgia, {USA}}, pages 249--261,
  1979.

\bibitem[Val82]{valiant1982reducibility}
L~Valiant.
\newblock Reducibility by algebraic projections in: Logic and algorithmic.
\newblock In {\em Symposium in honour of Ernst Specker}, pages 365--380, 1982.

\bibitem[VSBR83]{valiant1983fast}
Leslie~G. Valiant, Sven Skyum, Stuart Berkowitz, and Charles Rackoff.
\newblock Fast parallel computation of polynomials using few processors.
\newblock {\em SIAM Journal on Computing}, 12(4):641--644, 1983.

\bibitem[vzGG13]{von2013modern}
Joachim von~zur Gathen and J{\"u}rgen Gerhard.
\newblock {\em Modern computer algebra}.
\newblock Cambridge university press, 2013.

\bibitem[vzGK85]{von1985factoring}
Joachim von~zur Gathen and Erich Kaltofen.
\newblock Factoring sparse multivariate polynomials.
\newblock {\em Journal of Computer and System Sciences}, 31(2):265--287, 1985.

\bibitem[Zas69]{zassenhaus1969hensel}
Hans Zassenhaus.
\newblock {On Hensel factorization, I}.
\newblock {\em Journal of Number Theory}, 1(3):291--311, 1969.

\bibitem[ZS75]{zariski1975commutative}
Oscar Zariski and Pierre Samuel.
\newblock {\em Commutative algebra. II. Reprint of the 1960 edition},
  volume~29.
\newblock Graduate Texts in Mathematics, 1975.

\end{thebibliography}

\appendix

\section{Preliminaries}
\label{prel}

\subsection{Definition of ABP}
\label{ABP}

ABP is a {\em skew} circuit, i.e.~each multiplication gate has fanin two with at least one of its inputs being a variable or a field constant. A completely different definition can be given via layered graphs or iterated matrix multiplication or symbolic determinant. Famously, they are all equivalent up to polynomial blow up \cite{mahajan2014algebraic}.

\begin{definition}[Algebraic Branching Program]
An algebraic branching program (ABP) is a layered graph with a unique source vertex (say $s$) and a unique sink vertex (say $t$). All edges are from layer $i$ to $i+1$ and each edge is labelled by a linear polynomial. The polynomial computed by the ABP is defined as $f= \sum_{\gamma : s \leadsto t} \text{wt}(\gamma)$, where for every path $\gamma$ from $s$ to $t$, the weight wt$(\gamma)$ is defined as the product of the labels over the edges forming $\gamma$. 

{\em Size} of the ABP is defined as the total number of edges in the ABP. {\em Width} is the maximum number of vertices in a layer. 

Equivalently, one can define $f$ as a product of matrices (of dimension at most the width), each one having linear polynomials as entries. For more details, see \cite{shpilka2010arithmetic}.
\end{definition}

It is a famous result that the ABP model is the same as symbolic determinant \cite{mahajan1997combinatorial}.

\subsection{Randomized algorithm for linear algebra using PIT}

The following lemma from \cite{kopparty2015equivalence} discusses how to perform linear algebra when the coefficients of vectors are given as formula (resp.\ ABP). This will be crucially used in Theorem \ref{thm3} when we would give an algorithm to output the factors.

\begin{lemma}(Linear algebra using PIT \cite[Lem.2.6]{kopparty2015equivalence})
\label{lem-linsyst}
Let $M=(M_{i,j})_{k \times n}$ be a matrix (where $k$ is $n^{O(1)}$) with each entry being a degree $\leq n^{O(1)}$ polynomial in $\F[\overline{x}]$. Suppose, we have algebraic formula (resp.\ ABP) of size $\leq n^{O(\log n)}$ computing each entry. Then, there is a randomized poly($n^{\log n}$)-time algorithm that either:
\begin{itemize}
\item finds a formula (resp.\ ABP) of size poly$(n^{\log n})$  computing a non-zero $u \in (\F[\overline{x}])^n$ such that $Mu=0$, or
\item outputs $0$ which declares that $u=0$ is the only solution.
\end{itemize}
\end{lemma}

\begin{proof}
This was proved in \cite[Lem.2.6]{kopparty2015equivalence} for the circuit model. Since we are using a different model we repeat the details. The idea is the following. Iteratively, for every $r = 1,\hdots,n$ we shall find an $r \times r$ minor contained in the first $r$ columns that is full rank. While  continuing this process, we either reach $r = n$ in which case it means that the matrix has full column rank, hence, $u=0$ is the only solution, or we get stuck at some value say $r=r_0$. We use the fact that $r_0$ is rank and using this minor we construct the required non-zero vector $u$.

We explain the process in a bit more detail. Using a randomized algorithm, we look for some non-zero entry in the first column. If no such entry is found we can simply take $u = (1,0,\hdots,0)$. So assume that such a non-zero entry
is found. After permuting the rows we can assume wlog that this is $M_{1,1}$. Thus, we have found a $1\times1$ minor satisfying the requirements. Assume that we have found an $r \times r$ full rank minor that is composed of the first $r$ rows and columns (we can always rearrange and hence it can be assumed wlog that they correspond to first $r$ rows and columns). Denote this minor by $M_r$. 

Now for every $(r+1)\times (r+1)$ submatrix of $M$ contained in the first $r+1$ columns and containing $M_r$, we check whether the determinant is $0$ by randomized algorithm. If any of these submatrices have nonzero determinant, then we pick one of them and call it $M_{r+1}$. Otherwise, we have found that first $r + 1$ columns
of $M$ are linearly dependent. As $M_r$ is full rank, there is $v \in \F(\overline{x})^r$ such that $M_rv=(M_{1,r+1},\hdots,M_{r,r+1})^T$. This can be solved by applying Cramer's rule. The $i$-th entry of $v$ is of the form $\text{det}(M_r^{(i)})/\text{det}(M_r)$, where $M_r^{(i)}$ is obtained by replacing $i$-th column of $M_r$ with $(M_{1,r+1},\hdots,M_{r,r+1})^T$. Observe that det$(M_r)$, as well as det$(M_r^{(i)})$, are both in $\F[\overline{x}]$.

Then it is immediate that $u:= (\text{det}(M_r^{(1)}),\hdots,\text{det}(M_r^{(r)}),-\text{det}(M_r),0,\hdots,0)^T$ is the desired vector. 

To find $M_r$, each time we have to calculate the determinant and decide whether it is $0$ or not. This is simply PIT for a determinant polynomial with entries of algebraic complexity $n^{O(\log n)}$ and degree $n^{O(1)}$. So, we have a comparable randomized algorithm for this. Determinant of a symbolic $n \times n$ matrix has $n^{O(\log n)}$ size formula (resp.~poly$(n)$ ABP) \cite{mahajan1997combinatorial}. When the entries of the matrix have $n^{O(\log n)}$ size formula (resp.~ABP), altogether, the determinant polynomial has the same algebraic complexity. There are $<n^2$ PIT invocations to test zeroness of the determinant. Altogether, we have a poly($n^{\log n}$)-time randomized algorithm for this \cite{Sch80}.
\end{proof}   

\subsection{Basic operations on formula, ABP and circuit}

We use the following standard results on size bounds for performing some basic operations (like taking derivative) of circuits, formulas, ABPs. 
\begin{lemma}(Eliminate single division \cite{strassen1973vermeidung}, \cite[Thm.2.1]{shpilka2010arithmetic})
\label{div-elm}
Let $f$ and $g$ be two degree-$D$ polynomials, each computed by a circuit (resp.\ ABP resp.\ formula) of size-$s$ with $g(\overline{0}) \neq 0$. Then $f/g \bmod \langle \overline{x} \rangle^{d+1}$ can be computed by $O((s+d)d^3)$ (resp.~$O(sd^2D)$ resp.~$O(sd^2D^2)$) size circuit (resp.\ ABP resp.\ formula). 
\end{lemma}

\begin{proof}
Assume wlog that $g(\overline{0})=1$; we can ensure this by appropriate normalization. So, we have the following power series identity in $\F[[\overline{x}]]$:
$$f/g=f/(1-(1-g))=f+f(1-g)+f(1-g)^2+f(1-g)^3+ \cdots \,.$$
Note that this is a valid identity as $1-g$ is constant free. For all $d\ge0$,  LHS=RHS $ \bmod \langle \overline{x} \rangle^{d+1}$.

If we want to compute $f/g \bmod \langle \overline{x}\rangle^{d+1}$, we can take the RHS of the above identity up to the term $f(1-g)^d$ and discard the remaining terms of degree greater than $d$. The degree$>d$ monomials can be truncated, using Strassen's {\em homogenization} trick, in the case of circuits and ABPs (see \cite[Lem.5.2]{saptharishi2016survey}), and an {\em interpolation} trick in the case of formulas (which also works for ABPs and low degree circuits, \cite[Lem.5.4]{saptharishi2016survey}). A careful analysis shows that the size blow up is at most $O((s+d)d^2\cdot d)$ (resp.~$O(sd\cdot D\cdot d)$ resp.~$O(sd\cdot D^2\cdot d)$) for circuits (resp.\ ABP resp.\ formula).

Using the above result, it is easy to see, that we get poly$(s,d)$ size circuit (resp.\ ABP resp.\ formula) for computing $f/g \bmod \langle \overline{x} \rangle^{d+1}$.
\end{proof}

\noindent {\bf Remark. }
Note that it may happen that $g(\overline{0})=0$, thus $1/g$ does not exist in $\F[[\overline{x}]]$, yet $f/g$ may be a polynomial of degree $d$.  
In such a case, we need to discuss a modified {\em normalization} that works. We can shift the polynomials $f,g$ by some random $\overline{\alpha}\in\F^n$. The constant term of the shifted polynomial is non-zero with high probability \cite{Sch80}. Now, we compute $f(\overline{x}+\overline{\alpha})/g(\overline{x}+\overline{\alpha})$ using the method described above.  Finally, we recover the polynomial $f/g$ by applying the reverse shift $\overline{x}\mapsto {\overline{x}-\overline{\alpha}}$.

What if our model has several division gates?

\begin{lemma}(Div.~gates elimination \cite[Thm.2.12]{shpilka2010arithmetic}) \label{lem-div-elim}
Let $f$ be a polynomial computed by a circuit (resp.\ formula), using  division gates, of size $s$. Then, $f \bmod \langle \overline{x} \rangle^{d+1}$ can be computed by $\poly(sd)$ size circuit (resp.\ formula). 
\end{lemma}
\begin{proof}[Proof idea.]
We preprocess the circuit (resp.\ formula) so that the only division gate used in the modified circuit (resp.\ formula) is at the top. Now to remove the single division gate at the top, we use the above power series trick. 

The idea of the pre-processing is the following. We can separately keep track of numerator and denominator computed at each gate and simulate addition, multiplication and division gates in the original circuit. This pre-processing incurs only poly($sd$) blow up in the case of circuits. In the case of formulas one has to ensure that in any path from the leaf to the root, there are only $O(\log sd)$ division gates.
\end{proof}

\begin{lemma}[Derivative computation]
\label{lem-derivative}
If a polynomial $f(\overline{x},y)$ can be computed by a circuit (resp.\ formula resp.\ ABP) of size $s$ and degree $d$. Then, any $\frac{\partial^k f}{\partial y^k}$ can be computed by circuit (resp.\ formula resp.\ ABP) of size $\text{poly}(sk)$.
\end{lemma}

\begin{proof}
The idea is simply to use the homogenization and interpolation properties \cite[Sec.5.1-2]{saptharishi2016survey}.

Let $f(\overline{x},y)=c_0+c_1y+c_2y^2+\ldots+c_\delta y^\delta$, where $c_0,c_1,\ldots,c_\delta\in \F[\overline{x}]$.
Given the circuit (resp.\ formula resp.\ ABP) computing polynomial $f(\overline{x},y)$, we can get the circuits (resp.\ formula resp.\ ABP) computing $c_0,\ldots,c_\delta$ using homogenization and interpolation as discussed before.
Given $c_0,\ldots,c_\delta$, computing $\frac{\partial^k f}{\partial y^k}$ in size $\text{poly}(sd)$ is trivial. We use this approach of computing derivative when the polynomial is of degree $d\le\poly(s)$.

In the case of high degree circuits, we cannot use the above approach. \cite[Thm.1]{kaltofen1987single} shows that $\frac{\partial^k f}{\partial y^k}$ can be computed by a circuit of size $O(k^2s)$, i.e.~the degree of the circuit does not matter. The main idea is to inductively use the Leibniz product rule of $k$-th order derivative. 
\end{proof}

\subsection{Sylvester matrix \& resultant}\label{sec-res}

First, let us look at the notion of resultant of two univariate polynomials.
Let $p(x),q(x)\in \F[x]$ be of degree $a, b$ respectively. From Euclid's extended algorithm, it can be shown that there exist two polynomials $u(x),v(x)\in \F[x]$ such that $u(x)p(x)+ v(x)q(x)=\gcd(p(x),q(x))$. This is known as Bezout's identity. If $\gcd(p(x),q(x))=1$, then $(u,v)$ with $\deg(u)\leq b$ and $\deg(v)\leq a$ is unique. Let $u(x)=u_0+u_1x+u_2x^2+\ldots+u_{b}x^{b}$ and $v(x)=v_0+v_1x+\ldots+v_{a}x^{a}$. 

Now, if we use the equation $u(x)p(x)+ v(x)q(x)=\gcd(p(x),q(x))$ and compare the coefficients of $x^i$, for $0\le i\le a+b$, we get a system of linear equations in the $a+b+2$ many unknowns ($u_i$'s and $v_i$'s). The system of linear equations can be represented in the matrix form as $Mx=y$, where $x$ consists of the unknowns. 
Resultant of $f,g$ is defined as the determinant of the matrix $M$. It is easy to see that $M$ is invertible if and only if the polynomials are coprime.

Now, the notion of resultant can be extended to multivariate, by defining resultant of polynomials $f(\overline{x},y)$ and $g(\overline{x},y)$ wrt some variable $y$. The idea is same as before, now we take gcd wrt the variable $y$ and get a system of linear equations from Bezout's identity. The matrix can be explicitly written with entries being polynomial coefficients (or they could be from $\F[[\overline{x}]]$). This is known as Sylvester matrix, which we define next. 

\begin{definition}
\label{defn-sylvester}
Let $f(\overline{x},y)= \sum_{i=0}^l f_i(\overline{x})y^i$ and $g(\overline{x},y)=\sum_{i=0}^m g_i(\overline{x})y^i$. Define 
\emph{Sylvester} matrix of $f$ and $g$ wrt $y$ as the following $(m+l+1)\times(m+l+1)$ matrix: $$ \text{Syl}_y(f,g):=\, \begin{bmatrix}
f_l & 0 & 0 & \hdots & 0 & g_m & 0 & 0 & 0 \\
f_{l-1} & f_l & 0 & \hdots & 0 & g_{m-1} & g_m & 0 & 0 \\
f_{l-2} & f_{l-1} & f_l & \hdots & 0 & g_{m-2} & g_{m-1} & g_l & 0 \\
\vdots & \vdots &  \vdots &  \vdots & \vdots &  \vdots &  \vdots &  \vdots &  \vdots \\
f_0 & f_1 & \hdots & \hdots & f_l & g_0 & g_1 & \hdots & g_m \\
0 & f_0 & \hdots & \hdots & \hdots & 0 & g_0 & \hdots & 0 \\
\vdots & \vdots & \vdots & \vdots & \vdots & \vdots & \vdots & \vdots & \vdots \\
0 & \hdots & \hdots & \hdots & f_0 & 0 & \hdots & \hdots & g_0
\end{bmatrix}
$$
\end{definition}

So, resultant can be formally defined as follows (for more details and alternate definitions, see \cite[Chap.1]{LN}).

\begin{definition}
Given two polynomials $f(\overline{x},y)$ and $g(\overline{x},y)$, define the resultant of $f$ and $g$ wrt $y$ as determinant of the Sylvester matrix, 
$$ \text{Res}_y(f,g) \,:=\, \text{det}(\text{Syl}_y(f,g)) \,.$$

\end{definition}

From the definition, it can be seen that Res$_y(f,g)$ is a polynomial in $\F[x]$
with degree bounded by $2 \text{deg}(f) \text{deg}(g)$. Now, we state the following fundamental property of the Resultant, which is crucially used. 

\begin{proposition}[Res vs gcd]\label{prop-res}
 
\begin{enumerate}
\item Let $f , g \in \mathbb{F}[\overline{x}, y]$ be polynomials with positive degree in $y$. Then, $\text{Res}_y(f,g) = 0 \iff f $ and $g$
have a common factor in $\mathbb{F}[\overline{x}, y]$ which has positive degree in $y$.
\item There exists $u,v \in \F[\overline{x}]$ such that $uf+vg= \text{Res}_y(f,g)$.
\end{enumerate}
\end{proposition}

The proof of this standard proposition can be found in many standard books on algebra including \cite[Sec.6]{von2013modern}. 

\begin{lemma}[Squarefree-ness]
\label{lem-sq-free}
Let $f \in \mathbb{F}(\overline{x})[y]$ be a polynomial with deg$_y(f) \geq 1$. $f$ is square free iff  $f, f':=\partial_yf$ are coprime wrt $y$. 
\end{lemma}

\begin{proof}
The main idea is to show that there does not exist $g \in \mathbb{F}(\overline{x})[y]$ with positive degree in $y$ such that $g \mid \gcd_y(f(\overline{x},y), f'(\overline{x},y) )$. This is true because-- suppose $g$ is an irreducible polynomial with positive degree in $y$ that divides both $f(\overline{x},y)$ and $f'(\overline{x},y)$. So, 
$$f(\overline{x},y)=gh \implies f'(\overline{x},y)=gh'+g'h \implies g \mid g'h \,.$$ 

As $g$ is irreducible and deg$_y(g') < \text{deg}_y(g)$ we deduce that $g \mid h$. Hence, $g^2 \mid f$.  This contradicts the hypothesis that $f$ is square free. 
\end{proof}

Now, we state another standard lemma, which is useful to us and which is proved using the property of Resultant.

\begin{lemma}[Coprimality]
\label{lem-coprime}
Let $f, g \in \mathbb{F}(\overline{x})[y]$ be coprime polynomials wrt $y$ (\& nontrivial in $y$). Then, for $\overline{\beta} \in_r \mathbb{F}^{n}$, $f(\overline{\beta},y)$ and $g(\overline{\beta},y)$ are coprime (\& nontrivial in $y$). 
\end{lemma}

\begin{proof}
Consider $f= \sum_{i=1}^d f_iy^i$ and $g= \sum_{i=1}^e g_iy^i$. Choose a random $\overline{\beta} \in_r \F^n$. Then, by Proposition \ref{prop-res} \& \cite{Sch80}, $f_d\cdot g_e \cdot \text{Res}_y(f,g)$ at $\overline{x}=\overline{\beta}$ is nonzero. This in particular implies that $ \text{Res}_y(f(\overline{\beta},y), g(\overline{\beta},y)) \neq 0 $. 

This implies, by Proposition \ref{prop-res}, $f(\overline{\beta},y)$ and $g(\overline{\beta},y)$ are coprime. 
\end{proof}

\section{Useful in Section \ref{sec-split} }

\begin{lemma}(Power series root \cite[Thm.2.31]{burgisser2013algebraic})
\label{lem-NI}
Let $P(\overline{x},y) \in \F(\overline{x})[ y]$, $P'(\overline{x}, y)  = \frac{\partial P(\overline{x}, y)}{\partial y}$ and  $\mu \in \F$
be such that $P(\overline{0},\mu)  = 0$ but $P'(\overline{0},\mu)  \neq 0$ . Then, there is a unique power series $S$ such that $S(\overline{0})=\mu$ and $P(\overline{x},S)=0$ i.e. $$ y-S(\overline{x}) \mid P(\overline{x},y)\,.$$ 

Moreover, there exists a rational function $y_t$, $ \forall t \geq 0$, such that 
$$ y_{t+1} \,=\, y_t - \frac{P(\overline{x},y_t)}{P'(\overline{x},y_t)} \text{ and } S  \equiv y_t \bmod \langle \overline{x} \rangle ^{2^t} \text{ with } y_0=\mu \,.$$
\end{lemma}
\begin{proof}
 We give an inductive proof of existence and uniqueness together. Suppose $ P = \sum_{i=0}^d c_i y^i$. We show that there is $y_t$, a rational function $\frac{A_t}{B_t}$  such that $y_t \in \mathbb{F}[[\overline{x}]]$ , 
 For all $t \geq 0$,
 $
 P(\overline{x},y_t) \equiv 0 \bmod \langle \overline{x} \rangle ^{2^t}$ 
and for all $t \geq 1$,
 $ y_{t} \equiv y_{t-1} \bmod \langle \overline{x}\rangle ^{2^{t-1}} $. The proof is by induction. Let $y_0 := \mu$. Thus, base case is true. Now suppose such $y_t$ exists. Define  $y_{t+1}:= y_t - \frac{P(\overline{x},y_t)}{P'(\overline{x},y_t)} $.

 Now, $y_t \equiv y_{t-1} \bmod \langle \overline{x} \rangle ^{2^{t-1}} \implies y_t(\overline{0})=\mu $ . Hence  $P'(\overline{x},y_t) \rvert_{\overline{x}=\overline{0}} = P'(\overline{0},\mu) \neq 0$ and so $P'(\overline{x},y_t)$ is a unit in the power series ring. So, $y_{t+1}  \in \mathbb{F}[[\overline{x}]]$. Let us verify that it is an improved root of $P$; we use Taylor expansion. 
\begin{align*}
 P(\overline{x},y_{t+1})  & = P\left(\overline{x},\, y_t-\frac{P(\overline{x},y_t)}{P'(\overline{x},y_t)} \right) \\ 
 &= P(\overline{x},y_t) - P'(\overline{x},y_t)\frac{P(\overline{x},y_t)}{P'(\overline{x},y_t)}+ \frac{P''(\overline{x},y_t)}{2!} \left(\frac{P(\overline{x},y_t)}{P'(\overline{x},y_t)}\right)^2- \hdots \\ 
 & = 0 \, \bmod \langle \overline{x} \rangle ^{2^{t+1}} \,.
 \end{align*}
Thus, $P(\overline{x},y_{t+1}) \equiv 0 \bmod \langle \overline{x} \rangle ^{2^{t+1}} $ and $ y_{t+1} \equiv y_t \bmod \langle \overline{x} \rangle^{2^t}$. This completes the induction step.
 
Moreover, using the notion of limit, we have  $\lim_{t \to \infty}y_t = S$, a formal power series. It is unique as $\mu$ is a non-repeated root of 
$P(\overline{0},y)$. In particular, we get that for all $t \geq 0$, $P(\overline{x},S)=0$ or $y-S \mid P$.
\end{proof}

\begin{lemma}[Transform to monic]
For a polynomial $f(\overline{x})$ of total degree $d\ge0$ and random $\alpha_i \in_r \mathbb{F}$, the transformed polynomial $g(\overline{x},y):= f(\overline{\alpha}y+\overline{x})$ has a nonzero constant as coefficient of $y^d$, and degree wrt $y$ is $d$. 
\label{lem-monic}
\end{lemma}

\begin{proof}
Suppose the transformation is $x_i \mapsto x_i + \alpha_i y $ where $ i\in[n]$. Write $f = \sum_{|\overline{\beta}|=d} c_{\overline{\beta}} \overline{x}^{\overline{\beta}} + \text{ lower degree terms } $. Coefficient of $y^d$ in $g$ is $ \sum_{|\overline{\beta}|=d} c_{\overline{\beta}} \overline{\alpha}^{\overline{\beta}}$. Clearly, for a random $\overline{\alpha}$ this coefficient will not vanish \cite{Sch80}, and it is the highest degree monomial in $g$.  

This ensures $\text{deg}_y(g) = \text{deg}(f)=d$ and that $g$ is monic wrt $y$.
\end{proof}

\section{Useful in Section \ref{sec-main}}

\begin{lemma}[Matrix inverse]
\label{lem-inv-det}
Let $\mu_i, i\in[d]$, be distinct nonzero elements in $\F$. Define a $d\times d$ matrix $A$ with the $(i,j)$-th entry $1/(y_i-\mu_j)^2$. Its entries are in the function field $\F(\overline{y})$. 
Then, det$(A)\ne0$.
\end{lemma}

\begin{proof}
The idea is to consider the power series of the function $1/(y_i-\mu_j)^2$ and show that a monomial appears nontrivially in that of det$(A)$.

We first need a claim about the coefficient operator on the determinant.

\begin{claim}\label{clm-coef-det}
Let $f_j=\sum_{i\ge0} \beta_{j,i} x^i$ be a power series in $\F[[x]]$, for $j\in[d]$. Then, $\text{Coeff}_{\overline{x}^{\overline{\alpha}}}\,\circ \text{det}\left(f_j(x_i)\right)$ 
$=\, \text{det}\left(\beta_{j,\alpha_i}\right)$.
\end{claim}
\claimproof{clm-coef-det}{
Observe that the rows of the matrix have disjoint variables. Thus, $x_i^{\alpha_i}$ could be produced only from the $i$-th row. This proves:
$\text{Coeff}_{\overline{x}^{\overline{\alpha}}} \circ \text{det}\left(f_j(x_i)\right) \,=\, \text{det}\left(\text{Coeff}_{x_i^{\alpha_i}} \circ f_j(x_i)\right)
\,=\, \text{det}\left(\beta_{j,\alpha_i}\right) $.
}

By Taylor expansion we have 
$$ \frac{1}{(x-\mu)^2} \,=\, \frac{1}{\mu^2}\sum_{j\geq 1} j \left(\frac{x}{\mu}\right)^{j-1} \,.$$ 
Hence, the coefficient of $y_i^{i-1}$ in $A(i,j)$ is 
$$\frac{1}{\mu_j^2} \frac{i}{\mu_j^{i-1}}= \frac{i}{\mu_j^{i+1}} \,.$$ 

By the above claim, the coefficient of $\prod_{i\in[d]} y_i^{i-1}$ in det$(A)$ is:
det$\left(\left( \frac{i}{\mu_j^{i+1}} \right)\right)$. By cancelling $i$ (from each row) and $1/\mu_j^2$ (from each column), we simplify it to the Vandermonde determinant:
$$ \text{det}\begin{bmatrix}
\frac{1}{\mu_1^0} & \frac{1}{\mu_2^0} & \hdots & \frac{1}{\mu_d^0} \\ 
\frac{1}{\mu_1^1} & \frac{1}{\mu_2^1} & \hdots & \frac{1}{\mu_d^1} \\ 
\vdots & \vdots & \hdots & \vdots \\
\frac{1}{\mu_1^{d-1}} & \frac{1}{\mu_2^{d-1}} & \hdots & \frac{1}{\mu_d^{d-1}} 
\end{bmatrix} \,=\, \prod_{i < j\in[d]} \left({\frac{1}{\mu_i}-\frac{1}{\mu_j}}\right) \,\neq\, 0 \,.$$
 
Hence, the determinant of $A$ is non-zero. 
\end{proof}

{\bf Remark.} If the characteristic of $\F$ is a prime $p\ge2$ then the above proof needs a slight modification. One should consider the coefficient of $\prod_{i\in[d]} y_i^{s_i-1}$ in det$(A)$ for a set $S=\{s_1,\ldots,s_d\}$ of distinct non-negative integers that are not divisible by $p$.

\begin{lemma}[Series inverse]
\label{lem-series-id}
Let $\delta\ge1$.
Assume that $A$ is a polynomial of degree $<\delta$ and $B$ is a homogeneous polynomial of degree $\delta$, such that $A(\overline{0})=:\mu \neq 0$. Then, we have the following identity in $\F[[\overline{x}]](y)$:
$$\frac{1}{ y-(A+B)} \,\equiv\, \frac{1}{y-A} + \frac{B}{(y-\mu)^2} \,\bmod \langle  \overline{x} \rangle ^{\delta+1}$$ 
\end{lemma}
\begin{proof}
We will use the notation $A^{[1,\delta-1]}$ to refer to the sum of the homogeneous parts of $A$ of degrees between $1$ and $\delta-1$ (equivalently, it is $A^{<\delta}-\mu$). Note that $B\cdot A^{[1,\delta-1]}$ vanishes mod $\langle  \overline{x} \rangle ^{\delta+1}$. Now,
\begin{align*}
\frac{1}{ y-(A+B)}  & \equiv  \frac{1}{y-\mu-\left(A^{[1,\delta-1]}+B\right)} \bmod \langle  \overline{x} \rangle ^{\delta+1}\\ 
& \equiv \frac{1}{y-\mu} \left( \frac{1}{1-\frac{A^{[1,\delta-1]}+B}{y-\mu}} \right) \bmod \langle  \overline{x} \rangle ^{\delta+1} \\ 
& \equiv \frac{1}{y-\mu} \left( 1+ \left(\frac{A^{[1,\delta-1]}+B}{y-\mu}\right) +\left(\frac{A^{[1,\delta-1]}+B}{y-\mu}\right)^2+ \hdots \hdots  \right)\bmod \langle  \overline{x} \rangle ^{\delta+1} \\ & \equiv \frac{1}{y-\mu} \left( 1+ \left(\frac{A^{[1,\delta-1]}+B}{y-\mu}\right) +\left(\frac{A^{[1,\delta-1]}}{y-\mu}\right)^2 +\left(\frac{A^{[1,\delta-1]}}{y-\mu}\right)^3 + \hdots \hdots  \right)\bmod \langle  \overline{x} \rangle ^{\delta+1} \\ 
& \equiv \frac{1}{y-\mu} \left( 1+ \left(\frac{A^{[1,\delta-1]}}{y-\mu}\right) +\left(\frac{A^{[1,\delta-1]}}{y-\mu}\right)^2+ \hdots \hdots  \right)+ \frac{B}{(y-\mu)^2}\bmod \langle  \overline{x} \rangle ^{\delta+1} \\
& \equiv \frac{1}{y-\mu}\left(\frac{1}{1-\frac{A^{[1,\delta-1]}}{y-\mu}}\right)+ \frac{B}{(y-\mu)^2}\bmod \langle  \overline{x} \rangle  ^{\delta+1}\\
& \equiv \frac{1}{y-A} + \frac{B}{(y-\mu)^2}\bmod \langle  \overline{x} \rangle ^{\delta+1} \quad.
\end{align*}
\end{proof}

\subsection{Closure properties for VNP}
\label{VNP}

{\em VNP-size parameter} $(w,v)$ of $F$ refers to $w$ being the witness size and $v$ being the size of the verifier circuit $f$.

Let $F(\overline{x},y), G(\overline{x},y), H(\overline{x})$ have verifier polynomials $f,g,h$ and the VNP size parameters $(w_f,v_f), (w_g,v_g),(w_h,v_h)$ respectively. Let the degree of $F$ wrt $y$ be $d$. Then, the following closure properties can be shown 
(\cite{burgisser2013algebraic} or \cite[Thm.2.19]{burgisser2013completeness}):
\begin{enumerate}
    \item  Add (resp.~Multiply): $F+G$ (resp.~$FG$) has VNP-size parameter $(w_f+w_g, v_f+v_g+3)$.
    \item Coefficient: $F_i(\overline{x})$ has VNP-size parameter $(w_f, (d+1)(v_f+1))$, where $F(\overline{x}, y) =: \sum_{i=0}^d F_i(\overline{x})y^i$.
    \item Compose: $F(\overline{x},H(\overline{x}))$ has VNP-size parameter $((d+1)(w_f+ dw_h), (d+1)^2(v_f+v_h+1)) $.
\end{enumerate}

\begin{proof}

All the above statements are easy to prove using the definition of VNP.
\begin{enumerate}
\item $(FG)(\overline{x},y)= \left(\sum_{u \in \{0,1\}^{w_f}} f(\overline{x},u_1,\hdots,u_{w_f})\right)\cdot \left(\sum_{u \in \{0,1\}^{w_g}} g(\overline{x},u_1,\hdots,u_{w_g})\right)$ $=$\\ $\sum_{u \in \{0,1\}^{w_f+w_g}} A(\overline{x},u_1,\ldots,u_{w_f+w_g})$. Where, $A(\overline{x},u_1,\ldots,u_{w_f+w_g}):=$ $f(\overline{x},u_1,\hdots,u_{w_f})\cdot g(\overline{x},u_{w_f+1},\hdots,u_{w_f+w_g})$. Trivially, $A$ has size $v_f+v_g+3$ (extra: one node, two edges) and witness size is $w_f+w_g$. 
    
Similarly, with $F+G$.

\item Interpolation gives, $ f_i(\overline{x})= \sum_{j=0}^d \alpha_j F(\overline{x},\beta_j)$, for some distinct arguments $\beta_j\in \F$. Clearly, $F(\overline{x},\beta_j)$ has VNP-size parameter $(w_f,v_f)$. Using the previous addition property we get that the verifier circuit has size $(d+1)(v_f+1)$. Witness size remains $w_f$ as we can reuse the witness string of $F$.

\item Write $F(\overline{x},y)=: \sum_{i=0}^d F_i(\overline{x})y^i$. We know that $F_i$ has VNP-size parameter $(w_f, (d+1)(v_f+1))$. For $0\le i\le d$, $H^i$ has VNP-size parameter $(iw_h,(i+1)v_h)$ using $i$-fold product (Item 1). Substituting $y=H$ in $F$, we can calculate the VNP-size parameter. 

Suppose $F_i$ and $H^i$ have corresponding verifier circuits $A_i$ and $B_i$ respectively. Then, $F(\overline{x},H(\overline{x})) = \sum_{i=0}^d F_i(\overline{x})H^i(\overline{x}) = \sum_{i=0}^d \left(\sum_{u \in \{0,1\}^{w_f}} A_i(\overline{x},u)\right)\cdot \left(\sum_{u \in \{0,1\}^{iw_h}} B_i(\overline{x},u)\right)$. Thus, the witness size is $< (d+1)(w_f+ dw_h)$. The corresponding verifier circuit size is $< (d+1)^2(v_f+v_h+1)$. 
\end{enumerate}

\end{proof}





\end{document}